\documentclass[aps,11pt,twoside,notitlepage,superscriptaddress]{revtex4-1}

\usepackage{graphicx,epic,eepic,epsfig,amsmath,amsfonts,latexsym,amssymb,color,enumerate}

\def\squareforqed{\hbox{\rlap{$\sqcap$}$\sqcup$}}
\def\qed{\ifmmode\squareforqed\else{\unskip\nobreak\hfil
\penalty50\hskip1em\null\nobreak\hfil\squareforqed
\parfillskip=0pt\finalhyphendemerits=0\endgraf}\fi}
\def\endenv{\ifmmode\;\else{\unskip\nobreak\hfil
\penalty50\hskip1em\null\nobreak\hfil\;
\parfillskip=0pt\finalhyphendemerits=0\endgraf}\fi}

\newtheorem{theorem}{Theorem}

\newtheorem{corollary}[theorem]{Corollary}

\newtheorem{definition}[theorem]{Definition}

\newtheorem{lemma}[theorem]{Lemma}

\newtheorem{proposition}[theorem]{Proposition}
\newtheorem{remark}[theorem]{Remark}

\newenvironment{proof}[1][Proof]{\noindent\textbf{Proof.} }{\hfill\qed}
\newenvironment{proofof}[1][Proof]{\noindent\textbf{Proof~#1.} }{\hfill\qed}

\newcommand{\nc}{\newcommand}
\nc{\rnc}{\renewcommand}
\nc{\beq}{\begin{equation}}
\nc{\eeq}{\end{equation}}
\nc{\bsp}{\begin{split}}
\nc{\esp}{\end{split}}
\nc{\beqa}{\begin{eqnarray}}
\nc{\eeqa}{\end{eqnarray}}
\nc{\lbar}[1]{\overline{#1}}
\nc{\ket}[1]{|#1\rangle}
\nc{\bra}[1]{\langle#1|}
\nc{\braket}[2]{\langle #1 | #2 \rangle}
\nc{\ketbra}[2]{|#1\rangle\!\langle#2|}
\nc{\proj}[1]{| #1\rangle\!\langle #1 |}
\nc{\avg}[1]{\langle#1\rangle}
\nc{\Rank}{\operatorname{Rank}}
\nc{\smfrac}[2]{\mbox{$\frac{#1}{#2}$}}
\nc{\tr}{\operatorname{Tr}}
\nc{\ox}{\otimes}
\nc{\dg}{\dagger}
\nc{\dn}{\downarrow}
\nc{\supp}{{\operatorname{supp}}}
\nc{\qsupp}{{\operatorname{qsupp}}}
\nc{\var}{\operatorname{var}}
\nc{\rar}{\rightarrow}
\nc{\lrar}{\longrightarrow}
\nc{\polylog}{\operatorname{polylog}}
\nc{\1}{{\openone}}
\nc{\id}{{\operatorname{id}}}
\nc{\<}{\langle}
\nc{\Hom}[2]{\mbox{Hom}(\CC^{#1},\CC^{#2})}
\nc{\rU}{\mbox{U}}
\nc{\mc}{\mathcal}

\begin{document}

\title{Strong converse exponents of partially smoothed information measures}

\author{Mario Berta}
    \affiliation{Institute for Quantum Information, RWTH Aachen University, Aachen
                 52074, Germany}
\author{Yongsheng Yao}
	\email{yongsh.yao@gmail.com}
    \affiliation{Institute for Quantum Information, RWTH Aachen University, Aachen
                 52074, Germany}
\date{\today}


\begin{abstract}
Partially smoothed information measures are fundamental tools in one-shot quantum information theory. In this work, we determine the exact strong converse exponents of these measures for both pure quantum states and classical states. Notably, we find that the strong converse exponents based on trace distance takes different forms between pure and classical states, indicating that they are not uniform across all quantum states. Leveraging these findings, we derive the strong converse exponents for quantum data compression, intrinsic randomness extraction, and classical state splitting. A key technical step in our analysis is the determination of the strong converse exponent for classical privacy amplification, which is of independent interest.
\end{abstract}

\maketitle


\section{Introduction}

Partially smoothed information measures\,---\,including the partially smoothed mutual max-information and conditional min-entropy\,---\,were introduced in~\cite{ABJT2020partially} and have since become fundamental tools in quantum information theory. Unlike standard smoothed information measures~\cite{BCR2011the, TomamichelHayashi2013hierarchy}, which involve smoothing over all nearby states, partially smoothed measures require smoothing only over those nearby states whose marginal on a specified subsystem remains unchanged. This constraint aligns naturally with many quantum information processing tasks where certain subsystems are preserved, making these measures particularly effective for such scenarios. As a result, partially smoothed information measures often yield tighter bounds in one-shot information-theoretic tasks such as privacy amplification, quantum state merging, and channel simulation~\cite{ABJT2020partially, CRBT2024channel, FWTB2024channel, FBB2021thermodynamic,ADJ2017quantum}.

Like their standard counterparts, the first-order asymptotics of partially smoothed measures coincide with their corresponding mutual information or conditional entropy~\cite{ABJT2020partially}. However, unlike the well-understood second-order asymptotics and large-deviation behavior of standard smoothed measures~\cite{BCR2011the, TomamichelHayashi2013hierarchy}, much less is known about these refined asymptotics in the partially smoothed setting. Reference~\cite{ABJT2020partially} established the second-order asymptotics for the partially smoothed mutual max-information and conditional min-entropy\,---\,based on trace distance\,---\,in the fully classical case. Reference~\cite{AbdelhadiRenes2020second} extended this by analyzing the second-order asymptotics of the partially smoothed conditional min-entropy under purified distance for pure states. Moreover, it demonstrated that the second-order behavior differs for certain classes of states, indicating that the the second-order terms are not uniform across all states. These findings suggest that the fine-grained asymptotic analysis of partially smoothed information measures is significantly more intricate than for the standard case. Given their central role in one-shot quantum information theory, further study of their second-order asymptotics and large-deviation properties is essential for deriving sharper asymptotic characterizations of various quantum information processing tasks.

In this paper, we study the strong converse exponents of partially smoothed information measures. For a bipartite state $\rho_{RA}$, consider the smoothing quantities 
\[
\epsilon^\Delta_{\dot{R}:A}(\rho_{RA},\lambda):=\min \{\Delta(\rho_{RA},  \tilde{\rho}_{RA})~|\exists~ \text{quantum states}~\sigma_{A}, \tilde{\rho}_{RA}~\text{s.t.}~\tilde{\rho}_{R}=\rho_R, \tilde{\rho}_{RA} \leq 2^{\lambda} \rho_R \ox \sigma_{A}\}
\]
and
\[
\epsilon^\Delta_{A|\dot{R}}(\rho_{RA},\lambda):=\min \{\Delta(\rho_{RA},  \tilde{\rho}_{RA})~|\exists~ \text{quantum states}~\tilde{\rho}_{RA}~\text{s.t.}~\tilde{\rho}_{R}\leq \rho_R, \tilde{\rho}_{RA} \leq 2^{-\lambda} I_A \ox \rho_R\},
\]
where $\Delta$ can be either the trace distance or the purified distance and $I_A$ denotes the identity operator. When $r$ is fixed, as the 
number of $\rho_{RA}$ grows, the exponential decay rates under which $\epsilon^\Delta_{\dot{R^n}:A^n}(\rho^{\ox n}_{RA},nr)$ and $\epsilon^\Delta_{A^n|\dot{R^n}}(\rho^{\ox n}_{RA},nr)$ converge to $1$ are called the strong converse exponent for 
partially smoothing of the mutual max-information and the conditional min-entropy, respectively. As our main result, we establish the exact strong converse exponents of these partially smoothed information measures for classical states and pure states. 

For classical state $\rho_{RA}$, i.e., $\rho_{RA}$ is a bipartite probability distribution $\{p(x,a)\}_{(x,a)\in \mc{R} \times \mc{A}}$. Our main results are summarized in TABLE~\ref{tab:results} below.

\begin{table}[h!]
    \begin{center}
    \begin{tabular}{l|c|r}
        \hline
          & \textbf{Conditional Min-Entropy} & \textbf{Mutual Max-Information} \\
        \hline
        \textbf{Trace Distance}  & $ \sup\limits_{0\leq \alpha\leq 1}(1-\alpha)\big\{r-\bar{H}_{\alpha}(A|R)_p\big\}$ & $\sup\limits_{0 \leq \alpha \leq 1} (1-\alpha) \{I_{\alpha}(R:A)_p-r  \}$ \\
        \hline
        \textbf{Purified Distance}  & $\sup\limits_{\frac{1}{2} \leq \alpha \leq 1}\inf\limits_{t \in \mc{Q}(\mc{R})}
\big\{2D(t\|p)+\frac{1-\alpha}{\alpha}\big(r-\mathbb{E}_{x\sim t}H_\alpha(p(\cdot|x))\big)\big\}$ &  \\
        \hline
    \end{tabular}
\caption{Strong converse exponents of partially smoothed information measures for classical states. $\bar{H}_{\alpha}(A|R)_p$ is the Petz R\'enyi conditional entropy defined in Eq.~(\ref{def:con2}), $I_\alpha(R:A)_p$ is the Petz Rényi mutual information defined in Eq.~(\ref{def:mul1}), $\mc{Q}(\mc{R})$ denotes the set of all probability distributions on $\mc{R}$, $D(t\|p)$ is the relative entropy, $p(\cdot|x)$ is the marginal distribution on $\mc{A}$ conditional on $x$, and $H_{\alpha}(p(\cdot|x))$ is the R\'enyi entropy.}
\end{center}
    \label{tab:results}
\end{table}

\noindent For pure state $\rho_{RA}$ with Schmidt decomposition $\ket{\rho}_{RA}=\sum\limits_{x \in \mc{X}} \sqrt{p(x)}\ket{x}_R \ox \ket{x}_A$, We state our main results in TABLE~II.

\begin{table}[h!]
    \begin{center}
    \begin{tabular}{l|c|r}
        \hline
          & \textbf{Conditional Min-Entropy} & \textbf{Mutual Max-Information} \\
        \hline
        \textbf{Trace Distance}  & $ \inf\limits_{t \in \mc{Q}(\mc{X})} \big\{2D(t\|p)+|r+H(t)|^+  \big\}$ & $\sup\limits_{\beta>1} \frac{\beta-1}{\beta} \big\{2H_{\beta}(R)_\rho-r  \big\}$ \\
        \hline
        \textbf{Purified distance}  & $\inf\limits_{t \in \mc{Q}(\mc{X})} \big\{2D(t\|p)+|r+H(t)|^+  \big\}$ &  $\sup\limits_{\beta>1} \frac{\beta-1}{\beta} \big\{2H_{\beta}(R)_\rho-r  \big\}$ \\
        \hline
    \end{tabular}
\caption{Strong converse exponents of partially smoothed information measures for pure states. $H(t)$ is the Shannon entropy, $|x|^+$ denotes $\max\{x,0\}$, and $H_\beta(R)_\rho$ is the R\'enyi entropy of $\rho_R$. }
\end{center}
    \label{tab:results12}
\end{table}
\noindent In addition, we also show that when $\rho_{RA}$ is a pure state, the formulas of the strong converse exponents based on trace distance for classical states can not be converted to those for pure states. Hence, we find that the strong converse exponents of partially smoothed information measures based on trace distance are not uniform across states. As an application of these results, we derive the strong converse exponents of quantum data compression, intrinsic randomness and classical state splitting.

The core tools in establishing the strong converse exponents based on trace distance for classical states are two conclusions from~\cite[Proposition 2]{AbdelhadiRenes2020second} and~\cite[Lemma 2]{OCCB2024exponents}. In the derivation of the optimality parts of the strong converse exponents based on purified distance, we mainly use the operator H\"{o}lder inequality. Specifically, Lemma~\ref{lem:hof} proved in~\cite{WangWilde2019resource} plays an important role in our proofs. The establishment of the corresponding achievability parts depend on the relations between partially smoothed information measures and privacy amplification, quantum data compression. As an important ingredient in proving the achievability parts, we determine the strong converse exponent of classical privacy amplification which is of independent interest. The strong converse exponents based on trace distance for pure states can be obtained from those based on purified distance and an improved  Fuchs-van de Graaf inequality~(Lemma~\ref{lem:appen1} in Appendix).

The remainder of this paper is organized as follows. In Section~\ref{sec:preliminary}, we introduce some preliminaries in quantum information theory. In Section~\ref{sec:problem}, we present the main problems and the main results. In Section~\ref{sec:classical} , we give the proof of the strong converse exponents of partially smoothed information measures based on trace distance in classical setting. In Section~\ref{sec:donpuri}, we establish the strong converse exponent for partially smoothing of the conditional min-entropy based on purified distance. In Section~\ref{sec:constong} and Section~\ref{sec:strmutualpure} , we establish the strong converse exponents of partially smoothed information measures for pure states. In Section~\ref{sec:application}, we introduce the applications of the main results in quantum data compression, intrinsic randomness and classical state splitting. In Section~\ref{sec:conclu} , we conclude this paper with some discussion.


\section{Preliminaries}
\label{sec:preliminary}

\subsection{Notation}

In quantum information theory, every physical system is associated with a finite-dimensional Hilbert space $\mathcal{H}$. Let $\mc{L}(\mc{H})$ be the set of all linear operators on $\mc{H}$. We denote  the set of the  positive semi-definite  operators  on $\mc{H}$ as $\mc{P}(\mc{H})$. The set of normalized quantum states and sub-normalized quantum states on $\mc{H}$ are denoted as $\mathcal{S}(\mc{H})$ and $\mathcal{S}_{\leq}(\mc{H})$, respectively. We use the notation $|\mc{H}|$ for the dimension of $\mc{H}$. The identity operator on $\mc{H}$ is represented as $I_{\mc{H}}$. When $\mc{H}$ is associated with a system $A$, the above notations $\mc{L}(\mc{H})$, $\mc{P}(\mc{H})$, $\mathcal{S}(\mc{H})$, $\mathcal{S}_{\leq}(\mc{H})$, $|\mc{H}|$ and $I_{\mc{H}}$ also can be written as $\mc{L}(A)$, $\mc{P}(A)$, $\mathcal{S}(A)$, $\mathcal{S}_{\leq}(\mc{H})$, $|A|$ and $I_{A}$, respectively. The support of an operator $X$ is denoted by $\supp(X)$. For $A, B \in \mathcal{P}(\mathcal{H})$, we denote by $\{A \geq B\}$ the projection
onto the subspace spanned by the eigenvectors corresponding to the non-negative eigenvalues of $A-B$, $\{A>B\}$, $\{A \leq B\}$ and $\{A<B\}$ are defined similarly. We will use a famous inequality~\cite[Theorem 1]{AKCMBMA2007discriminating} stated as follows.
Let $A, B \in \mathcal{P}(\mathcal{H})$ and $s \in [0,1]$, then,
\begin{equation}
\label{equ:discri}
\tr A\{A \leq B\}+\tr B \{A>B\} \leq \tr A^{1-s}B^s.
\end{equation}
For a finite alphabet set $\mc{X}$, we denote as $\mc{Q}(\mc{X})$ and $|\mc{X}|$ the set of all probability distributions on $\mc{X}$ and the size of $\mc{X}$, respectively. For a bipartite probability distribution $P_{XY}\in \mc{Q}(\mc{X}\mc{Y})$, we let $P_{XY}(\cdot|x)$  represent the 
marginal distribution on $\mc{Y}$ conditional on $x$.

The purified distance between two quantum states $\rho, \sigma \in \mc{S}(\mc{H})$ is given by
\[
P(\rho,\sigma):=\sqrt{1-F^2(\rho,\sigma)},
\]
where $F(\rho,\sigma):=\|\sqrt{\rho}\sqrt{\sigma}\|_1$ is the fidelity function. The trace distance between $\rho$ and $\sigma$ is defined as
\[
d(\rho,\sigma):=\frac{1}{2}\|\rho-\sigma\|_1.
\]

A quantum channel $\mc{N}_{A \rightarrow B}$ is a completely positive and trace-preserving linear map from $\mc{L}(A)$ to $\mc{L}(B)$.  Let $A$ be a self-adjoint operator with spectral projections $P_1, \ldots, P_r$. The pinching channel $\mathcal{E}_A$ associated with $A$ is defined as
\[
\mathcal{E}_A : X\mapsto\sum_{i=1}^r P_i X P_i.
\]
The pinching inequality~\cite{Hayashi2002optimal} tells that for any $\sigma \in \mathcal{P}(\mathcal{H})$, we have
\begin{equation}
\sigma \leq v(A) \mathcal{E}_A(\sigma),
\end{equation}
where $v(A)$ is the number of different eigenvalues of $A$.

For $S_n$ being the symmetric group of the permutations of $n$ elements, we define  the set of symmetric states on $A^n$ as
\beq
\mc{S}_{\rm{sym}}(A^n):=\left\{\sigma_{A^n} \in \mc{S}(A^n)~|~
W^{\pi}_{A^n} \sigma_{A^n}W^{\pi * }_{A^n}=\sigma_{A^n}, \ \forall\ \pi \in S_n\right\}.
\eeq
where $W^{\pi}_{A^n}: \ket{\psi_1} \ox \ldots \ox \ket{\psi_n}
\mapsto  \ket{\psi_{\pi^{-1}(1)}} \ox \ldots \ox \ket{\psi_{\pi^{-1}(n)}}$ is the natural representation of $\pi \in S_n$. There exists a  symmetric state that dominates all the other symmetric states, as stated in the following Lemma~\ref{lem:sym}.  The paper~\cite{Hayashi2009universal} and~\cite{CKR2009postselection} give two different constructions of this symmetric state, respectively. See~\cite[Lemma 1]{HayashiTomamichel2016correlation} and~\cite[Appendix A]{MosonyiOgawa2017strong} for detailed arguments.
\begin{lemma}
\label{lem:sym}
For a finite-dimensional system $A$ and any $n\in\mathbb{N}$, there exists a  symmetric state $\sigma_{A^n}^u \in \mc{S}_{\rm{sym}}(A^n)$ such that every symmetric state $\sigma_{A^n} \in \mc{S}_{\rm{sym}}(A^n)$ is dominated as
\begin{equation}
\sigma_{A^n} \leq v_{n,|A|}\sigma_{A^n}^u,
\end{equation}
where $v_{n,|A|} \leq (n+1)^{\frac{(|A|+2)(|A|-1)}{2}}$ is a polynomial of $n$. The number of distinct eigenvalues of $\sigma_{A^n}^u$ is upper bounded by $v_{n,|A|}$ as well.
\end{lemma}


\subsection{Quantum entropies}

For $\rho , \sigma \in \mc{P}(\mc{H})$, the quantum relative
entropy of $\rho$ and $\sigma$ is defined as~\cite{Umegaki1954conditional}
\begin{equation}
D(\rho\|\sigma):= \begin{cases}
\tr(\rho(\log\rho-\log\sigma)) & \text{ if }\supp(\rho)\subseteq\supp(\sigma), \\
+\infty                        & \text{ otherwise.}
                  \end{cases}
\end{equation}
Let $\rho_{AB} \in \mc{S}(AB)$, the conditional entropy and mutual information
of $\rho_{AB}$ are defined, respectively as
\begin{align}
H(A|B)_\rho:&= -\min_{\sigma_B \in \mc{S}(B)} D(\rho_{AB} \| I_A \ox \sigma_B) \\
&=-D(\rho_{AB} \|I_A \ox \rho_B), \\
I(A:B)_\rho:&=\min_{\sigma_B \in \mc{S}(B)} D(\rho_{AB} \| \rho_A \ox \sigma_B).
\end{align}

The Petz R\'enyi divergence and the sandwiched R\'enyi divergence were introduced in~\cite{Petz1986quasi} and~\cite{MDSFT2013on, WWY2014strong}, respectively, which occupy an important position among various quantum R\'enyi divergences. Since they appeared, they have been widely used to 
characterize the large-deviation exponential rates of quantum information processing tasks~\cite{CHDH2020non, MosonyiOgawa2017strong, LiYao2024operational, MosonyiOgawa2015quantum, GuptaWilde2015multiplicativity, CMW2016strong, LiYao2024strong, LYH2023tight, Hayashi2015precise, LiYao2021reliable, HayashiTomamichel2016correlation}. 

\begin{definition}
\label{definition:sand}
Let $\alpha\in(0,+\infty)\setminus\{1\}$, $\rho\in\mc{S}(\mc{H})$ and $\sigma\in\mc{P}(\mc{H})$.
When $\alpha >1$ and $\supp(\rho)\subseteq\supp(\sigma)$ or $\alpha\in (0,1)$ and $\supp(\rho)\not\perp\supp(\sigma)$, the Petz R\'enyi divergence $D_{\alpha}(\rho\|\sigma)$ and the sandwiched R{\'e}nyi divergence $D_\alpha^*(\rho\|\sigma)$
are defined as
\begin{align}
&D_{\alpha}(\rho \| \sigma):=\frac{1}{\alpha-1} \log Q_{\alpha}(\rho \| \sigma)-\frac{1}{\alpha-1} \log \tr \rho,
\quad\text{with}\ \
Q_{\alpha}(\rho \| \sigma)=\tr \rho^\alpha \sigma^{1-\alpha}; \\
&D_{\alpha}^*(\rho \| \sigma):=\frac{1}{\alpha-1} \log Q_{\alpha}^*(\rho \| \sigma)-\frac{1}{\alpha-1} \log \tr \rho,
\quad\text{with}\ \
Q_{\alpha}^*(\rho \| \sigma)=\tr {({\sigma}^{\frac{1-\alpha}{2\alpha}} \rho {\sigma}^{\frac{1-\alpha}{2\alpha}})}^\alpha;
\end{align}
otherwise, we set $D_{\alpha}^*(\rho \| \sigma)=+\infty$. 
\end{definition}
When $\alpha$ tends to $1$, $D_\alpha(\rho\|\sigma)$ and $D_\alpha^*(\rho\|\sigma)$ converge to the quantum relative entropy and when $\alpha$ goes to infinity,  $D_\alpha^*(\rho\|\sigma)$ converges to the max-relative entropy~\cite{Datta2009max}
\begin{equation}
D_{\rm{max}}(\rho\|\sigma):=\inf\{\lambda~|~\rho \leq 2^\lambda \sigma\}.
\end{equation}

Similar to the definition of the conditional entropy, the Petz conditional R\'enyi entropy~\cite{TBH2014relating} and the  sandwiched conditional R\'enyi entropy~\cite{MDSFT2013on} of order $\alpha\in(0,+\infty)\setminus\{1\}$ for a bipartite state $\rho_{AB}$ are defined as
\begin{align}
H_\alpha(A|B)_\rho:&=-\inf_{\sigma_B \in \mc{S}(B)} D_{\alpha}(\rho_{AB}\|I_A \ox \sigma_B), \\ \label{def:con2}
\bar{H}_\alpha(A|B)_\rho:&= D_{\alpha}(\rho_{AB}\|I_A \ox \rho_B), \\ 
H^*_\alpha(A|B)_\rho:&=-\inf_{\sigma_B \in \mc{S}(B)} D^*_{\alpha}(\rho_{AB}\|I_A \ox \sigma_B), \\
\bar{H}^*_\alpha(A|B)_\rho:&= D_{\alpha}(\rho_{AB}\|I_A \ox \rho_B).
\end{align}
Likewise, the Petz R\'enyi mutual information~\cite{HayashiTomamichel2016correlation} and the sandwiched R\'enyi mutual information~\cite{WWY2014strong, Beigi2013sandwiched, MckinlayTomamichel2020decomposition}  of
order $\alpha \in (0,+\infty)\setminus\{1\}$ for $\rho_{AB} \in \mc{S}(AB)$ are defined as
\begin{align}
\label{def:mul1}
I_\alpha(A:B)_\rho:&=\inf_{\sigma_B \in \mc{S}(B)} D_{\alpha}(\rho_{AB}\|\rho_A \ox \sigma_B), \\ 
I^*_\alpha(A:B)_\rho:&=-\inf_{\sigma_B \in \mc{S}(B)} D^*_{\alpha}(\rho_{AB}\|\rho_A \ox \sigma_B).
\end{align}

In the following proposition, we collect some important properties of these quantum R\'enyi divergences.
\begin{proposition}
\label{prop:mainpro}
For $\rho ,\sigma \in \mc{P}(\mc{H})$, the Petz R\'enyi divergence and the sandwiched R{\'e}nyi
divergence satisfy the following properties.
\begin{enumerate}[(i)]
  \item Monotonicity in $\sigma$~\cite{MDSFT2013on,MosonyiOgawa2017strong}: if $\sigma' \geq \sigma$, then $D^{(t)}_{\alpha}(\rho \| \sigma') \leq D^{(t)}_{\alpha}(\rho \| \sigma)$,
      for $(t)=\{\}$, $\alpha \in [0,+\infty)$ and $(t)=*$,  $\alpha \in [\frac{1}{2},+\infty)$;
   \item  Data processing inequality\cite{Petz1986quasi, FrankLieb2013monotonicity, Beigi2013sandwiched, MDSFT2013on, WWY2014strong, MosonyiOgawa2017strong}: for any quantum channel $\mc{N}$ from $\mc{L}(\mc{H})$ to $\mc{L}(\mc{H}')$, we have
      \begin{equation}
      D_{\alpha}^{(\rm{t})}(\mc{N}(\rho) \| \mc{N}(\sigma)) \leq D_{\alpha}^{(\rm{t})}(\rho \| \sigma),
      \end{equation}
      for $(t)=\{\}$, $\alpha \in [0,2]$ and $(t)=*$, $\alpha \in [\frac{1}{2},+\infty)$;
 \item Duality relation~\cite{HayashiTomamichel2016correlation}: for pure state $\rho_{ABC}\in \mc{S}(ABC)$ and $\tau_A \in \mc{P}(A)$ such that $\supp(\rho_A) \subseteq \supp(\tau_A)$, we have
  \begin{equation}
 \begin{split}
\inf_{\sigma_B \in \mc{S}(B)}D_\alpha^*(\rho_{AB} \|\tau_A \ox \sigma_B)&=- \inf_{\sigma_C \in \mc{S}(C)}D_\beta^*(\rho_{AC} \|\tau^{-1}_A \ox \sigma_C),~\text{for}~\alpha \in [\frac{1}{2},+\infty],~\frac{1}{\alpha}+\frac{1}{\beta}=2,\\
 \inf_{\sigma_B \in \mc{S}(B)}
D_\alpha(\rho_{AB} \|\tau_A \ox \sigma_B)&=-D_\beta^*(\rho_{AC} \|\tau^{-1}_A \ox \rho_C),~\text{for}~\alpha \in [\frac{1}{2},+\infty],~\beta=\frac{1}{\alpha},\\
D_\alpha(\rho_{AB} \|\tau_A \ox \rho_B)&=-D_\beta(\rho_{AC} \|\tau^{-1}_A \ox \rho_C),~\text{for}~\alpha \in [0,2],~\alpha+\beta=2;
 \end{split}
 \end{equation}
\item Approximation by pinching~\cite{MosonyiOgawa2015quantum,HayashiTomamichel2016correlation}:
        for $\alpha\geq 0$, we have
        \begin{equation}
        D^*_\alpha(\mc{P}_\sigma(\rho)\|\sigma)
        \leq D^*_\alpha(\rho\|\sigma)
        \leq D^*_\alpha(\mc{P}_\sigma(\rho)\|\sigma)+2\log v(\sigma);
        \end{equation}
\item Convexity in $\alpha$~\cite{MosonyiOgawa2017strong}:  the function $\alpha \mapsto\log Q^{(t)}_\alpha(\rho \| \sigma)$ is convex on $(0, +\infty)$ for $t=\{\}$ and $t=*$;
\item Convexity in $\sigma$~\cite{MosonyiOgawa2017strong}: the function $\sigma \mapsto D_{\alpha}^{(\rm{t})}(\rho \| \sigma)$ is convex on $\mc{P}(\mc{H})$ for $t=\{\}$, $\alpha \in [0,2]$ and $t=*$, $\alpha \in [\frac{1}{2}, +\infty)$;
\item Continuity in $\alpha$~\cite{MosonyiOgawa2017strong}:  the function $\alpha \mapsto\log Q^{(t)}_\alpha(\rho \| \sigma)$ is continuous on $[0, +\infty]$ for $t=\{\}$ and $t=*$;
\item Lower semi-continuity in $\sigma$~\cite{MosonyiOgawa2017strong}: the function $\sigma \mapsto D_{\alpha}^{(\rm{t})}(\rho \| \sigma)$ is lower semi-continuous on $\mc{P}(\mc{H})$ for $t=\{\}$, $\alpha \in (0,\infty) \setminus \{1\}$ and $t=*$, $\alpha \in (0,\infty) \setminus \{1\}$;
\item Additivity of the quantum R\'enyi mutual information~\cite{HayashiTomamichel2016correlation}:
        for $\rho_{AB}\in\mc{S}(AB)$, $\sigma_{A'B'} \in \mc{S}(A'B')$, we have
        \begin{equation}
        I_{\alpha}^{(\rm{t})}(AA':BB')_{\rho \ox \sigma}=I_{\alpha}^{(\rm{t})}(A:B)_\rho+I_{\alpha}^{(\rm{t})}(A':B')_\sigma,
        \end{equation}
   for $(t)=\{\}$, $\alpha \in [0,+\infty)$ and $(t)=*$, $\alpha \in [\frac{1}{2},+\infty)$;  
\item  Monotonicity under discarding classical information~\cite{LWD2016strong}:
        for the state $\sigma_{XAB}$ that is classical on $X$ and for $\alpha\in
        (0,+\infty)$, we have
        \begin{equation}
           \bar{H}^*_{\alpha}(AX|B)_\sigma \geq \bar{H}^*_{\alpha}(A|B)_\sigma;
        \end{equation}
\item  Variational expression~\cite{MosonyiOgawa2017strong}:
      when $\rho$ commutes with $\sigma$, we have
      \begin{equation}
          D_{\alpha}(\rho \| \sigma)= \begin{cases}
         \min\limits_{\tau \in \mc{S}_\rho(\mc{H})} \big\{D(\tau \| \sigma)
         -\frac{\alpha}{\alpha-1}D(\tau \| \rho)\big\}, & \alpha \in (0,1), \\
         \max\limits_{\tau \in \mc{S}_\rho(\mc{H})} \big\{D(\tau \| \sigma)
         -\frac{\alpha}{\alpha-1}D(\tau \| \rho)\big\}, & \alpha \in (1,+\infty),
      \end{cases} 
      \end{equation}
      where $\mc{S}_\rho(\mc{H}):=\{\tau~|~\tau \in \mc{P}(\mc{H}),~\supp(\tau) \subseteq \supp(\rho),~\tau~commutes~with~
      \rho~and~\sigma\}$.
\end{enumerate}
\end{proposition}


\subsection{Methods of types}

The method of types is a basic tool in information theory. We introduce the relevant definitions and properties in this section. We refer interested readers to~\cite{CsiszarKorner2011information}  for an overall introduction.

Let $\mc{X}$ be a finite alphabet set. For a sequence $x^n \in \mc{X}^{\times n}$, the type $t_{x^n}$ is the empirical distribution of $x^n$, i.e.,
\beq
t_{x^n}(a)=\sum_{i=1}^n \frac{\delta_{x_i,a}}{n}, \quad \forall a \in \mc{X}.
\eeq
We use the notation $\mc{T}_n^{\mc{X}}$ to represent the set of all types. The size of $\mc{T}_n^{\mc{X}}$ can be bounded as
\begin{equation}
\label{eq:typenumber}
|\mc{T}_n^{\mc{X}}| \leq (n+1)^{|\mc{X}|}.
\end{equation}
If $t\in\mc{T}_n^{\mc{X}}$, then we denote the set of all the sequences of type $t$ as $T_n^t$, that is,
\beq
T_n^t:=\{x^n~|~t_{x^n}=t\}.
\eeq
The size of  $T_n^t$ satisfies the following relation:
\begin{equation}
\label{eq:numt}
(n+1)^{-|\mc{X}|} 2^{nH(t)} \leq |T_n^t| \leq 2^{nH(t)},
\end{equation}
where $H(t)=-\sum\limits_{x\in\mc{X}} t(x)\log t(x)$ is the Shannon entropy of $t$. Let $X_1,X_2,\ldots,X_n$ be a sequence of i.i.d.\ random variables each taking values in $\mc{X}$ according to a distribution $p$. Then the probability of $X^n$ being equal to $x^n$ of type $t$ is given by
\begin{equation}
\label{eq:prot}
p^n(x^n)=\prod_{i=1}^np(x_i)=2^{-nH(t)-nD(t\|p)},
\end{equation}
where $D(t\|p)=\sum\limits_{x\in\mc{X}}t(x)\log\frac{t(x)}{p(x)}$ is the relative entropy. The combination of Eq.~(\ref{eq:numt}) and Eq.~(\ref{eq:prot}) gives that the probability of $X^n$ taking values in  $T_n^t$, given by $\sum\limits_{x^n\in T_n^t}p^n(x^n)$, can be bounded as
\begin{equation}
\label{eq:proset}
(n+1)^{-|\mc{X}|}2^{-nD(t\|p)} \leq \sum_{x^n\in T_n^t}p^n(x^n) \leq 2^{-nD(t\|p)}.
\end{equation}


\section{Main results}
\label{sec:problem}

For $\rho_{RA} \in \mc{S}(RA)$ and $\epsilon \in [0,1]$, the partially smoothed mutual max-information~\cite{ABJT2020partially} and the partially smoothed conditional min-entropy~\cite{ABJT2020partially} are defined, respectively as 
\begin{align}
\begin{split}
I_{\rm{max}}^{\epsilon, \Delta}(\dot{R}:A)_\rho&:=\inf_{\tilde{\rho}_{RA} \in \mc{G}^{\epsilon, \Delta}(\rho_{RA})} \inf_{\sigma_{A} \in \mc{S}(A)} D_{\rm{max}}(\tilde{\rho}_{RA} \| \rho_R \otimes \sigma_A),\\
H_{\rm{min}}^{\epsilon, \Delta}(A|\dot{R})_\rho&:=-\inf_{\tilde{\rho}_{RA} \in \mc{F}^{\epsilon, \Delta}(\rho_{RA})} D_{\rm{max}}(\tilde{\rho}_{RA} \| I_A \otimes \rho_R),
\end{split}
\end{align}
where $\Delta$ can be chosen as the trace distance or the purified distance and
\begin{align}
\begin{split}
&\mc{G}^{\epsilon, \Delta}(\rho_{RA}):=\{\tilde{\rho}_{RA} \in \mc{S}(RA)~|~\Delta(\rho_{RA}, \tilde{\rho}_{RA})\leq \epsilon, \tilde{\rho}_R=\rho_R \},\\
&\mc{F}^{\epsilon, \Delta}(\rho_{RA}):=\{\tilde{\rho}_{RA} \in \mc{S}_{\leq}(RA)~|~\Delta(\rho_{RA}, \tilde{\rho}_{RA})\leq \epsilon, \tilde{\rho}_R\leq \rho_R \}.
\end{split}
\end{align}
They are important tools in one-shot quantum information theory and have been used to give tight characterizations for the optimal performances of several quantum information processing tasks~\cite{ABJT2020partially, CRBT2024channel, FWTB2024channel, FBB2021thermodynamic, ADJ2017quantum}.

Regarding $I_{\rm{max}}^{\epsilon, \Delta}(\dot{R}:A)_\rho$ and $H_{\rm{min}}^{\epsilon, \Delta}(A|\dot{R})_\rho$ as functions of $\epsilon$, their inverse functions are given by
\begin{align}
\begin{split}
\epsilon^{\Delta}_{\dot{R}:A}(\rho_{RA},\lambda):&=\min\{ \epsilon~|~I_{\rm{max}}^{\epsilon, \Delta}(\dot{R}:A)_\rho \leq \lambda\} \\
&= \min\{ \Delta(\tilde{\rho}_{RA}, \rho_{RA})~|~\tilde{\rho}_{RA} \in \mc{S}(RA), (\exists \sigma_A \in \mc{S}(A)),\tilde{\rho}_{RA}\leq 2^\lambda \rho_R \otimes \sigma_A, \tilde{\rho}_{R}=\rho_R   \} ,\\
\epsilon^{\Delta}_{A|\dot{R}}(\rho_{RA},\lambda):&=\min\{ \epsilon~|~H_{\rm{min}}^{\epsilon, \Delta}(A|\dot{R})_\rho \geq \lambda\}  \\
&= \min\{ \Delta(\tilde{\rho}_{RA}, \rho_{RA})~|~\tilde{\rho}_{RA} \in \mc{S}_{\leq}(RA), \tilde{\rho}_{RA}\leq 2^{-\lambda} I_A \otimes \rho_R, \tilde{\rho}_{R}\leq \rho_R   \}.
\end{split}
\end{align} 
The reference~\cite{ABJT2020partially}  has derived the first-order asymptotics for $I_{\rm{max}}^{\epsilon, \Delta}(\dot{R}:A)_\rho$ and $H_{\rm{min}}^{\epsilon, \Delta}(A|\dot{R})_\rho$, i.e,
\begin{equation}
\label{equ:asyeq}
\begin{split}
\lim_{n \rightarrow \infty}\frac{1}{n} I_{\rm{max}}^{\epsilon, \Delta}(\dot{R^n}:A^n)_{\rho^{\ox n}} &=I(R:A)_\rho, \\
\lim_{n \rightarrow \infty}\frac{1}{n} H_{\rm{min}}^{\epsilon, \Delta}(A^n|\dot{R^n})_{\rho^{\ox n}} &=H(A|R)_\rho.
\end{split}
\end{equation}
Eq.~(\ref{equ:asyeq}) implies that 
when $r < I(R:A)_\rho$~($r>H(A|R)_\rho$), $\epsilon_{\dot{R^n}:A^n}^\Delta(\rho_{RA}^{\otimes n}, nr)$~($\epsilon_{A^n|\dot{R^n}}^\Delta(\rho_{RA}^{\otimes n}, nr)$) converges to $1$ exponentially fast. The exact rate of this exponential convergence is called the strong converse exponent for partially smoothing of the  mutual max-information~(partially smoothing of the conditional min-entropy).

In this paper, we establish the exact strong converse exponents of these two partially smoothed information measures for classical states and pure states.  Our main results are stated below.
\begin{theorem}
\label{thm:main1}
Let $\rho_{RA} \in \mc{S}(RA)$ be a classical state, i.e., $\rho_{RA}$ can be regarded as a bipartite probability distribution $\{p(x,a)\}_{(x,a)\in \mc{R} \times \mc{A}}$.  For any $r \in \mathbb{R}$, we have
\begin{align}
\label{equ:mainl}
\lim_{n \rightarrow \infty} \frac{-1}{n} \log(1-\epsilon^d_{A^n|\dot{R^n}}(\rho_{RA}^{\otimes n}, nr)) &= \sup_{0 \leq \alpha \leq 1}
(1-\alpha)\big\{r-\bar{H}_{\alpha}(A|R)_p\big\}, \\
\label{equ:clamutual}
\lim_{n \rightarrow \infty} \frac{-1}{n} \log(1-\epsilon^d_{\dot{R^n}:A^n}(\rho_{RA}^{\otimes n}, nr)) &= \sup_{0 \leq \alpha \leq 1}
(1-\alpha)\big\{I_\alpha(R:A)_p-r\big\}, \\
\label{equ:claconditional}
\lim_{n \rightarrow \infty} \frac{-1}{n} \log(1-\epsilon^P_{A^n|\dot{R^n}}(\rho_{RA}^{\otimes n}, nr)) &= \sup_{\frac{1}{2} \leq \alpha \leq 1}\inf_{t \in \mc{Q}(\mc{R})}
\big\{2D(t\|p)+\frac{1-\alpha}{\alpha}\big(r-\mathbb{E}_{x\sim t}H_\alpha(p(\cdot|x))\big)\big\}.
\end{align}
\end{theorem}

\begin{theorem}
\label{thm:main}
For any pure state $\ket{\rho_{RA}}=\sum\limits_{x \in \mc{X}} \sqrt{p(x)}\ket{x}_R \ox \ket{x}_A$, $r \in \mathbb{R}$ and $\Delta \in \{d, P\}$, we have
\begin{align}
\label{equ:main}
\lim_{n \rightarrow \infty} \frac{-1}{n} \log(1-\epsilon^\Delta_{A^n|\dot{R^n}}(\rho_{RA}^{\otimes n}, nr))&= \inf_{t \in \mc{Q}(\mc{X})} \big\{2D(t\|p)+|r+H(t)|^+  \big\}, \\
 \label{equ:puremutual}
\lim_{n \rightarrow \infty} \frac{-1}{n} \log(1-\epsilon_{\dot{R^n}:A^n}^\Delta(\rho_{RA}^{\otimes n}, nr))&= \sup_{\beta>1}
\frac{\beta-1}{\beta}\big\{2H_{\beta}(R)_\rho-r\big\},
\end{align}
\end{theorem}
where we use $|x|^+$ to represent $\max\{x, 0\}$.

\begin{remark}
From Lemma~\ref{lem:rela} in Appendix, if $\rho_{RA}$ is a pure state, we have
\begin{equation}
\label{equ:exeui}
\begin{split}
&\epsilon^d_{A^n|\dot{R^n}}(\rho_{RA}^{\otimes n}, nr) \leq \epsilon^P_{A^n|\dot{R^n}}(\rho_{RA}^{\otimes n}, nr) \leq \sqrt{\epsilon^d_{A^n|\dot{R^n}}(\rho_{RA}^{\otimes n}, nr)}, \\
&\epsilon^d_{\dot{R^n}:A^n}(\rho_{RA}^{\otimes n}, nr) \leq \epsilon^P_{\dot{R^n}:A^n}(\rho_{RA}^{\otimes n}, nr) \leq \sqrt{\epsilon^d_{\dot{R^n}:A^n}(\rho_{RA}^{\otimes n}, nr)}. 
\end{split}
\end{equation}
Eq.~(\ref{equ:exeui}) implied that 
\begin{equation}
\begin{split}
1-\epsilon^d_{A^n|\dot{R^n}}(\rho_{RA}^{\otimes n}, nr) \geq 1-\epsilon^P_{A^n|\dot{R^n}}(\rho_{RA}^{\otimes n}, nr) \geq  \frac{1-\epsilon^d_{A^n|\dot{R^n}}(\rho_{RA}^{\otimes n}, nr)}{1+\sqrt{\epsilon^d_{A^n|\dot{R^n}}(\rho_{RA}^{\otimes n}, nr)}}, \\
1-\epsilon^d_{\dot{R^n}:A^n}(\rho_{RA}^{\otimes n}, nr) \geq 1-\epsilon^P_{\dot{R^n}:A^n}(\rho_{RA}^{\otimes n}, nr) \geq  \frac{1-\epsilon^d_{\dot{R^n}:A^n}(\rho_{RA}^{\otimes n}, nr)}{1+\sqrt{\epsilon^d_{\dot{R^n}:A^n}(\rho_{RA}^{\otimes n}, nr)}}.
\end{split}
\end{equation}
Hence, the strong converse exponents based on trace distance and those based on purified distance are the same. In the following, we only need to establish
Theorem~\ref{thm:main} for purified distance.
\end{remark}

\begin{remark}
The readers might conjecture that Eq.~(\ref{equ:mainl}) and Eq.~(\ref{equ:clamutual}) hold for any quantum states. However, we prove in Appendix that Eq.~(\ref{equ:mainl}) and Eq.~(\ref{equ:clamutual}) can not be converted to Eq.~(\ref{equ:main}) and Eq.~(\ref{equ:puremutual}) when $\rho_{RA}$ is a pure state, thus denying this conjecture.
\end{remark}


\section{Partially smoothed classical information measures based on trace distance}
\label{sec:classical}

In this section, we establish the strong converse exponents of partially smoothed information measures based on  trace distance for classical states, i.e., Eq.~(\ref{equ:mainl}) and Eq.~(\ref{equ:clamutual}) in Theorem~\ref{thm:main1}. 

\begin{proofof}[of Eq.~(\ref{equ:mainl}) and Eq.~(\ref{equ:clamutual})]
Firstly, we prove Eq.~(\ref{equ:mainl}). It has been proved in~\cite[Proposition 2]{AbdelhadiRenes2020second} that for any classical state $\rho_{RA}$, we have
\begin{equation}
\label{equ:equreal}
H_{\rm{min}}^{\epsilon, d}(A|\dot{R})_{\rho}=H_{\rm{min}}^{\epsilon, d}(A|R)_{\rho},
\end{equation}
where $H_{\rm{min}}^{\epsilon, d}(A|R)_{\rho}:=-\inf\limits_{\tilde{\rho}_{RA} \in \mc{S}_{\leq}(RA):d(\rho_{RA},\tilde{\rho}_{RA})\leq \epsilon} D_{\rm{max}}(\tilde{\rho}_{RA} \|I_A\ox \rho_R)$ is the standard smoothed conditional min-entropy. We can obtain from Eq.~(\ref{equ:equreal}) that
\begin{equation}
\label{equ:global}
\begin{split}
&\epsilon^{d}_{A|\dot{R}}(\rho_{RA},\lambda) \\
=&\min\{ \epsilon~|~H_{\rm{min}}^{\epsilon, \Delta}(A|\dot{R})_\rho \geq \lambda\}  \\
=&\min\{ \epsilon~|~H_{\rm{min}}^{\epsilon, \Delta}(A|R)_\rho \geq \lambda\}  \\
=&\min\{ d(\tilde{\rho}_{RA}, \rho_{RA})~|~\tilde{\rho}_{RA} \in \mc{S}_{\leq}(RA), \tilde{\rho}_{RA}\leq 2^{-\lambda} I_A \otimes \rho_R \}.
\end{split}
\end{equation}
Theorem~2 in~\cite{SalzmannDatta2022total} and Eq.~(\ref{equ:global}) give  
\begin{equation}
\lim_{n \rightarrow \infty} \frac{-1}{n} \log(1-\epsilon^d_{A^n|\dot{R^n}}(\rho_{RA}^{\otimes n}, nr)) =\sup_{0 \leq \alpha \leq 1}
(1-\alpha)\big\{r-\bar{H}_{\alpha}(A|R)_p\big\}.
\end{equation}

Now, we start the proof of Eq.~(\ref{equ:clamutual}). 
We first deal with $``\leq"$ part. We write $\rho_{RA}$ in the form of  classical-quantum states for making the proof easier to read, i.e., $\rho_{RA}=\sum\limits_{x\in \mc{R}}p(x) \proj{x}_R \ox \rho_{A}^x$,  where $\{\ket{x}\}$ is an orthonormal
basis of the underlying Hilbert space $\mc{H}_R$ and $\rho_A^x=\sum\limits_{a\in \mc{A}}p(a|x)\proj{a}$, $\forall~x \in \mc{R}$. We denote the set of quantum states in $\mc{S}(A)$ which have the same eigenvectors with $\{\rho_A^x\}$ as $\mc{S}^{\rho}_{\rm{cla}}(A)$. Then, for any $\sigma_A \in \mc{S}^{
\rho}_{\rm{cla}}(A)$, we have
\begin{equation}
\label{equ:clafirst}
\begin{split}
&\epsilon^d_{\dot{R^n}:A^n}(\rho^{\ox n}_{RA},nr) \\
\leq &\sum_{x^n} p^n(x^n) \inf_{\tilde{\rho}^{x^n} \in \mc{S}(A^n): \tilde{\rho}^{x^n} \leq 2^{nr}\sigma_A^{\ox n}} d(\tilde{\rho}^{x^n},\rho_{A^n}^{x^n})\\
=&\sum_{x^n} p^n(x^n) \tr(\rho_{A^n}^{x^n}-2^{nr}\sigma_A^{\ox n})_+ \\
=&\tr(\rho_{RA}^{\ox n}-2^{nr}\rho_R^{\ox n} \ox \sigma_A^{\ox n})_+ \\
\leq & \tr \rho^{\ox n}_{RA}\{ \rho^{\ox n}_{RA} \geq  2^{nr}\rho_R^{\ox n} \ox \sigma_A^{\ox n}  \},
\end{split}
\end{equation}
where the first equality is from~\cite[Lemma 2]{OCCB2024exponents}. Eq.~(\ref{equ:clafirst}) and~\cite[Remark 4.6]{HMOerror2008}  imply that
\begin{equation}
\label{equ:intermi}
\lim_{n \rightarrow \infty}\frac{-1}{n} \log (1-\epsilon^d_{\dot{R^n}:A^n}(\rho^{\ox n}_{RA},nr)) \leq \sup_{0 \leq \alpha \leq 1} (1-\alpha)\{D_{\alpha}(\rho_{RA}\|\rho_R \ox \sigma_A)-r   \}.
\end{equation}
 Noticing that Eq.~(\ref{equ:intermi}) holds for any $\sigma_A \in \mc{S}^{
 \rho}_{\rm{cla}}(A)$, we have
\begin{equation}
\begin{split}
\lim_{n \rightarrow \infty}\frac{-1}{n} \log (1-\epsilon^d_{\dot{R^n}:A^n}(\rho^{\ox n}_{RA},nr)) &\leq \inf_{\sigma_A \in \mc{S}_{\rm{cla}}(A)}\sup_{0 \leq \alpha \leq 1} (1-\alpha)\{D_{\alpha}(\rho_{RA}\|\rho_R \ox \sigma_A)-r \} \\
&=\sup_{0 \leq \alpha \leq 1} (1-\alpha)\{I_\alpha(R:A)_p-r  \},
\end{split}
\end{equation}
where in the second equality, we use Sion's minimax theorem, the convexity~(concavity) and lower semi-continuous (upper semi-continuous) that this theorem require can be found in Proposition~\ref{prop:mainpro}~(\romannumeral5)-(\romannumeral8).

Next, we turn to the proof of  $``\geq"$ part. 
Because $\rho^{\ox n}_{RA}$ is a symmetric state, it is easy to verify that the infimum in $\epsilon^d_{\dot{R^n}:A^n}(\rho^{\ox n}_{RA},nr)$ can be restricted to over the symmetric states.
Hence, we can evaluate $\epsilon^d_{\dot{R^n}:A^n}(\rho^{\ox n}_{RA},nr)$ as
\begin{equation}
\label{equ:claconve}
\begin{split}
&\epsilon^d_{\dot{R^n}:A^n}(\rho^{\ox n}_{RA},nr)\\
=&\min\{ d(\tilde{\rho}_{R^nA^n}, \rho^{\ox n}_{RA})~|~\tilde{\rho}_{R^nA^n} \in \mc{S}(R^nA^n), (\exists \sigma_{A^n} \in \mc{S}_{\rm{sym}}(A^n)),\tilde{\rho}_{R^nA^n}\leq 2^{nr} \rho^{\ox n}_{R} \otimes \sigma_{A^n}, \tilde{\rho}_{R^n}= \rho^{\ox n}_{R}   \} \\
\geq & \min_{\sigma_{A^n} \in \mc{S}_{\rm{sym}}(A^n)} \tr(\rho_{RA}^{\ox n}-2^{nr}\rho_{R}^{\ox n}\ox \sigma_{A^n})_+ \\
\geq& \tr(\rho_{RA}^{\ox n}-v_{n,|A|}2^{nr}\rho_{R}^{\ox n}\ox \sigma^u_{A^n})_+,
\end{split}
\end{equation}
where the last line is from Lemma~\ref{lem:sym}. From Eq.~(\ref{equ:claconve}) and Eq.~(\ref{equ:discri}), for any $\alpha \in [0,1]$, we have
\begin{equation}
\label{equ:lain}
\begin{split}
& 1-\epsilon^d_{\dot{R^n}:A^n}(\rho_{RA}^{\ox n},nr) \\
\leq & \tr \rho_{RA}^{\ox n} \{\rho_{RA}^{\ox n} \leq v_{n,|A|}2^{nr}\rho_{R}^{\ox n}\ox \sigma^u_{A^n} \}+v_{n,|A|}2^{nr}\tr \rho_{R}^{\ox n}\ox \sigma^u_{A^n}\{\rho_{RA}^{\ox n} >v_{n,|A|}2^{nr}\rho_{R}^{\ox n}\ox \sigma^u_{A^n}\} \\
\leq & v^{1-\alpha}_{n,|A|}2^{nr(1-\alpha)} \tr \big(\rho_{RA}^{\ox n}\big)^\alpha \big(\rho_{R}^{\ox n}\ox \sigma^u_{A^n} \big)^{1-\alpha}.
\end{split}
\end{equation}
 Eq.~(\ref{equ:lain}) gives 
\begin{equation}
\begin{split}
&\lim_{n \rightarrow \infty} \frac{-1}{n} \log (1-\epsilon^d_{\dot{R^n}:A^n}(\rho_{RA}^{\ox n},nr) ) \\
\geq & (1-\alpha)\{ \lim_{n \rightarrow \infty}\frac{D_\alpha(\rho_{RA}^{\ox n}\| \rho_{R}^{\ox n}\ox \sigma^u_{A^n} )}{n}-r  \} \\
\geq & (1-\alpha)\{ \lim_{n \rightarrow \infty} \inf_{\sigma_{A^n}\in \mathcal{S}(A^n)}\frac{D_\alpha(\rho_{RA}^{\ox n}\| \rho_{R}^{\ox n}\ox \sigma_{A^n} )}{n}-r  \} \\
=&(1-\alpha)\{ I_{\alpha}(R:A)_p-r  \},
\end{split}
\end{equation}
where the equality is from Proposition~\ref{prop:mainpro}~(\romannumeral9). The above formula holds for any $\alpha \in [0,1]$, hence we have
\begin{equation}
\lim_{n \rightarrow \infty} \frac{-1}{n} \log (1-\epsilon^d_{\dot{R^n}:A^n}(\rho_{RA}^{\ox n},nr) ) \geq \sup_{0 \leq \alpha \leq 1} (1-\alpha)\{ I_{\alpha}(R:A)_p-r  \}.
\end{equation}
\end{proofof}


\section{Partially smoothed classical conditional min-entropy based on purified distance}
\label{sec:donpuri}

In this section, our objective is to determine the strong converse exponent for partially smoothing of the conditional min-entropy based on purified distance for classical states, i.e., Eq.~(\ref{equ:claconditional}) in Theorem~\ref{thm:main1}. We begin this procedure by first establishing the strong converse exponent of classical privacy amplification which serves as a key tool to prove the achievability part of Eq.~(\ref{equ:claconditional}).


\subsection{Classical privacy amplification}
\label{subsec:privacy}

For a classical-quantum state
\[
\rho_{RA}=\sum_{a \in \mc{A}} q(a) \rho_{R}^a \ox \proj{a}_A,
\]
where the system $A$ is a classical random variable that is partially correlated with an adversary's system $R$. In the procedure of privacy amplification, we apply a hash function $f : \mc{A} \mapsto \mc{Z}$ to extract a random number $Z$, which  is expected to be uniformly distributed and independent of the adversary’s system $R$. The action of the hash function $f$ can be expressed as a quantum
channel
\[
\Phi_f : \omega_A \mapsto \sum_{a\in \mc{A}} \ket{f(a)}\bra{a} \omega_A \ket{a}\bra{f(a)}.
\]
So the resulting state is $\Phi_f(\rho_{RA})$, the size  and the performance of this protocol are characterized by $|\mc{Z}|$ and the purified distance between the resulting state and the ideal state $\rho_R\ox \frac{I_Z}{|\mc{Z}|}$, respectively.

The reference~\cite{DevetakWinter2005distillation} has proved that for  asymptotically perfect privacy amplification, a rate of randomness 
extraction equals to $H(A|R)_\rho$ is necessary and sufficient. For finer asymptotic analysis, including 
second-order expansion  based on purified distance~\cite{TomamichelHayashi2013hierarchy} and that based on trace distance~\cite{SGC2022optimal}, as well as error exponent~\cite{LYH2023tight}, have been determined later.

When the rate of randomness extraction is larger than $H(A|R)_\rho$, the strong converse property holds, i.e., the optimal performance among all protocols with size $2^{nr}$ converges to $1$ as $n$ goes to $\infty$. The 
exact exponent of this decay is called the strong converse exponent and is defined as
\[
E^{\rm{pa}}_{sc}(\rho_{RA},r):=\lim_{n \rightarrow \infty} \frac{-1}{n} \log (1-\inf_{f_n\in \mc{F}_n(r)}P(\Phi_{f_n}(\rho_{RA}^{\ox n}),\rho_R^{\ox n} \ox \frac{I_{Z^n}}{|\mc{Z}^n|})),
\]
where $\mc{F}_n(r)$ is the set of functions from $\mc{A}^{\times n}$ to $\mc{Z}_n=\{1, \ldots, 2^{nr}\}$.

The reference~\cite{SalzmannDatta2022total, SGC2022strong} derived some lower bounds for $E^{\rm{pa}}_{sc}(\rho_{RA},r)$. However, the exact expression for $E^{\rm{pa}}_{sc}(\rho_{RA},r)$ is still unknown. In this subsection, we establish the exact 
strong converse exponent when the adversary's system $R$ is classical. The result is stated as follows.
\begin{theorem}
\label{thm:privacy}
Let $\rho_{RA}$ be the classical state defined in Theorem~\ref{thm:main1} and $r \geq 0$. Then, we have
\begin{equation}
\label{equ:privacy}
E^{\rm{pa}}_{sc}(\rho_{RA},r)= \sup_{\frac{1}{2} \leq \alpha \leq 1}\inf_{t \in \mc{Q}(\mc{R})}
\big\{2D(t\|p)+\frac{1-\alpha}{\alpha}\big(r-\mathbb{E}_{x\sim t}H_\alpha(p(\cdot|x))\big)\big\}.
\end{equation}
\end{theorem}

In the following, We accomplish the proof of Theorem~\ref{thm:privacy} by establishing the optimality part and the achievability part, respectively.

\begin{proofof}[of the optimality part]
As we did in the proof of Eq.~(\ref{equ:clamutual}), we write $\rho_{RA}$ in the form of classical-quantum states~($\rho_{RA}=\sum\limits_{x \in \mc{R}}p(x)\proj{x}_R\ox \rho_A^x$) in the rest of this section for convenience of reading the proofs.

For any $f_n \in \mc{F}_n(r)$ and $\frac{1}{2}<\alpha<1$, we can upper bound $F(\Phi_{f_n}(\rho^{\ox n}_{RA}), \rho_{R}^{\ox n} \ox \frac{I_{Z_n}}{|\mc{Z}_n|}) $ as follows.
\begin{equation}
\label{equ:priupp}
\begin{split}
&F(\Phi_{f_n}(\rho^{\ox n}_{RA}), \rho_{R}^{\ox n} \ox \frac{I_{Z_n}}{|\mc{Z}_n|}) \\
= &\sum_{x^n \in\mc{R}^{\times n}} p^n(x^n) F(\Phi_{f_n}(\rho_{A^n}^{x^n}), \frac{I_{Z_n}}{|\mc{Z}_n|}) \\
\leq &\sum_{x^n \in\mc{R}^{\times n}} p^n(x^n) 2^{\frac{1-\alpha}{2\alpha}(-D_\alpha(\Phi_{f_n}(\rho_{A^n}^{x^n})\| \frac{I_{Z_n}}{|\mc{Z}_n|})+D_\beta( \frac{I_{Z_n}}{|\mc{Z}_n|}\| \frac{I_{Z_n}}{|\mc{Z}_n|}))}\\
\leq &\sum_{x^n \in\mc{R}^{\times n}} p^n(x^n) 2^{\frac{1-\alpha}{2\alpha}(-D_\alpha(\rho_{A^n}^{x^n} \|	I_{A}^{\ox n})-nr)}  \\
=&\sum_{t \in \mc{T}_n^{\mc{R}}} (\sum_{x^n \in T_n^t}p^n(x^n))2^{\frac{1-\alpha}{2\alpha}(-n\mathbb{E}_{x \sim t} D_\alpha(\rho_{A}^x\| I_{A})-nr)} \\
\leq &(n+1)^{|R|} \sup_{t \in \mc{T}_n^\mc{R}} 2^{-nD(t\|p)} 2^{\frac{1-\alpha}{2\alpha}(-n\mathbb{E}_{x \sim t}D_\alpha(\rho_{A}^x\| I_{A})-nr)},
\end{split}
\end{equation}
where the first inequality is due to Lemma~\ref{lem:hof} in Appendix, the second inequality comes from Proposition~\ref{prop:mainpro}~(\romannumeral10) and the last inequality is from Eq.~(\ref{eq:proset}). Eq.~(\ref{equ:priupp}) implies that
\begin{equation}
\label{equ:prifire}
\begin{split}
&E^{\rm{pa}}_{sc}(\rho_{RA},r)\\
=&\lim_{n \rightarrow \infty} \frac{-1}{n} \log (1-\inf_{f_n\in \mc{F}_n(r)}P(\Phi_{f_n}(\rho_{RA}^{\ox n}),\rho_R^{\ox n} \ox \frac{I_{Z_n}}{|\mc{Z}_n|})) \\
\geq &\lim_{n\rightarrow \infty} \frac{-1}{n} \log \sup_{f_n \in \mc{F}_n(r)} F^2(\Phi_{f_n}(\rho^{\ox n}_{RA}), \rho_{R}^{\ox n} \ox \frac{I_{Z_n}}{|\mc{Z}_n|}) \\
\geq &\inf_{t \in \mc{Q}(\mc{R})}
\big\{2D(t\|p)+\frac{1-\alpha}{\alpha}(\mathbb{E}_{x\sim t}D_\alpha(\rho_A^x \|I_A)+r)\big\} \\
=&\inf_{t \in \mc{Q}(\mc{R})} \big\{2D(t\|p)+\frac{1-\alpha}{\alpha}\big(r-\mathbb{E}_{x\sim t}H_\alpha(p(\cdot|x))\big)\big\}.
\end{split}
\end{equation}
Noticing that Eq.~(\ref{equ:prifire}) holds for any $\frac{1}{2}<\alpha<1$, we complete the proof.
\end{proofof}

Before we prove the achievability part, we establish a variational expression for the formula at the right hand side of Eq.~(\ref{equ:privacy}). 
\begin{proposition}
\label{prop:var}
For $\rho_{RA}$ and $r$ stated in Theorem~\ref{thm:privacy}, it holds that
\begin{equation}
\begin{split}
 &\sup_{\frac{1}{2} \leq \alpha \leq 1}\inf_{t \in \mc{Q}(\mc{R})}\big\{2D(t\|p)+\frac{1-\alpha}{\alpha}\big(r-\mathbb{E}_{x\sim t}H_\alpha(p(\cdot|x))\big)\big\}\\
 =&\inf_{t \in \mc{Q}(\mc{R})}\inf_{\{\tau_A^x\}\in \mc{F}}
\big\{2D(t\|p)+\mathbb{E}_{x \sim t} D(\tau_A^x\|\rho_A^x)+|r+\mathbb{E}_{x \sim t}D(\tau_A^x\|I_A)|^+ \big\},
\end{split}
\end{equation}
where $\mc{F}:=\{\{\tau_A^x\}_{x \in \mc{R}}~|~\tau_A^x \in \mc{S}_{\rho_{A}^x},~\forall x \in \mc{R}\}$.
\end{proposition}

\begin{proof}
By direct calculation, we have
\begin{equation}
\begin{split}
 &\sup_{\frac{1}{2} \leq \alpha \leq 1}\inf_{t \in \mc{Q}(\mc{R})}\big\{2D(t\|p)+\frac{1-\alpha}{\alpha}\big(r-\mathbb{E}_{x\sim t}H_\alpha(p(\cdot|x))\big)\big\}\\
 =&\sup_{\frac{1}{2} \leq \alpha \leq 1}\inf_{t \in \mc{Q}(\mc{R})}\big\{2D(t\|p)+\frac{1-\alpha}{\alpha}(\mathbb{E}_{x \sim t}D_\alpha(\rho_A^x \|I_A)+r)\big\} \\
=&\sup_{\frac{1}{2} < \alpha < 1}\inf_{t \in \mc{Q}(\mc{R})}\big\{2D(t\|p)+\frac{1-\alpha}{\alpha}(\mathbb{E}_{x \sim t}D_\alpha(\rho_A^x \|I_A)+r)\big\} \\
=&\sup_{\frac{1}{2} < \alpha < 1}\inf_{t \in \mc{Q}(\mc{R})}\inf_{\{\tau_A^x\}\in \mc{F}} \big\{2D(t\|p)+\frac{1-\alpha}{\alpha}(\mathbb{E}_{x \sim t}D(\tau_A^x \|I_A)+\frac{\alpha}{1-\alpha}\mathbb{E}_{x \sim t}D(\tau_A^x \|\rho_A^x)+r)\big\} \\
=&\sup_{\frac{1}{2} < \alpha < 1}\inf_{t \in \mc{Q}(\mc{R})}\inf_{\{\tau_A^x\}\in \mc{F}} \big\{2D(t\|p)+\mathbb{E}_{x \sim t}D(\tau_A^x \|\rho_A^x)+\frac{1-\alpha}{\alpha}(\mathbb{E}_{x \sim t}D(\tau_A^x \|I_A)+r)\big\} \\
=&\sup_{0 < s < 1}\inf_{t \in \mc{Q}(\mc{R})}\inf_{\{\tau_A^x\}\in \mc{F}} \big\{2D(t\|p)+\mathbb{E}_{x \sim t}D(\tau_A^x \|\rho_A^x)+s(\mathbb{E}_{x \sim t}D(\tau_A^x \|I_A)+r)\big\} \\
=&\inf_{t \in \mc{Q}(\mc{R})}\inf_{\{\tau_A^x\}\in \mc{F}}\sup_{0 < s < 1} \big\{2D(t\|p)+\mathbb{E}_{x \sim t}D(\tau_A^x \|\rho_A^x)+s(\mathbb{E}_{x \sim t}D(\tau_A^x \|I_A)+r)\big\} \\
=&\inf_{t \in \mc{Q}(\mc{R})}\inf_{\{\tau_A^x\}\in \mc{F}}\big\{2D(t\|p)+\mathbb{E}_{x \sim t}D(\tau_A^x\|\rho_A^x)+|r+\mathbb{E}_{x \sim t}D(\tau_A^x\|I_A)|^+ \big\},
\end{split}
\end{equation}
where the third line follows from Proposition~\ref{prop:mainpro}~(\romannumeral11) and in the sixth line, we use Sion's minimax theorem.
\end{proof}

With the above variational expression in hand, we can begin the proof of the achievability part.

\begin{proofof}[of the achievability part]
By proposition~\ref{prop:var}, we only need to prove that for any $t\in \mc{Q}(\mc{R})$ and $\{\tau_A^x\}_{x \in \mc{R}} \in \mc{F}$,
\begin{equation}
\label{equ:prva}
    E^{\rm{pa}}_{sc}(\rho_{RA},r) \leq 2D(t\|p)+\mathbb{E}_{x\sim t} D(\tau_A^x\|\rho_A^x)+|r+\mathbb{E}_{x\sim t}D(\tau_A^x\|I_A)|^+ .
\end{equation}
We divide the proof of Eq.~(\ref{equ:prva}) into two cases according to whether $r$ is larger or lower than $-\mathbb{E}_{x\sim t}D(\tau_A^x\|I_A)$.

\textbf{Case 1:}  $r<-\mathbb{E}_{x\sim t}D(\tau_A^x\|I_A)$;

Because $r<-\mathbb{E}_{x\sim t}D(\tau_A^x\|I_A)$, $r$ is an achievable rate of randomness extraction from $\tau_{RA}:=\sum\limits_{x \in \mc{R}}t(x) \proj{x}_R \ox \tau_A^x$~\cite{DevetakWinter2005distillation}.
Hence, \cite[Theorem~1]{Hayashi2015precise} tells us that there exists a sequence of hash functions $\{f_n~: \mc{A}^{\times n} \rightarrow \mc{Z}_n=\{1,\ldots,2^{nr}\}\}_{n \in \mathbb{N}}$ such that 
\begin{equation}
\label{equ:f3}
\lim_{n \rightarrow \infty}D(\Phi_{f_n}(\tau_{RA}^{\ox n}), \tau_R^{\ox n} \ox \frac{I_{Z_n}}{|\mc{Z}_n|})=\lim_{n \rightarrow \infty}\sum_{x^n \in  \mc{R}^{\times n}}t^n(x^n) D(\Phi_{f_n}(\tau_{A^n}^{x^n})\| \frac{I_{Z_n}}{|\mc{Z}_n|})=0.
\end{equation}
Eq.~(\ref{equ:f3}) implies that 
\begin{equation}
\label{equ:l4}
\begin{split}
&1\\
=&\lim_{n \rightarrow \infty} 2^{-\sum\limits_{x^n \in  \mc{R}^{\times n}}t^n(x^n) D(\Phi_{f_n}(\tau_{A^n}^{x^n})\| \frac{I_{Z_n}}{|\mc{Z}_n|})} \\
=&\lim_{n \rightarrow \infty} 2^{-\sum\limits_{s \in \mc{T}_n^{\mc{R}}} t^n(T_n^s) \frac{1}{|T_n^s|}\sum\limits_{x^n \in  T_n^s} D(\Phi_{f_n}(\tau_{A^n}^{x^n})\| \frac{I_{Z_n}}{|\mc{Z}_n|})} \\
\leq &\lim_{n \rightarrow \infty}\sum\limits_{s \in \mc{T}_n^{\mc{R}}} t^n(T_n^s) 2^{-\frac{1}{|T_n^s|}\sum\limits_{x^n \in  T_n^s} D(\Phi_{f_n}(\tau_{A^n}^{x^n})\| \frac{I_{Z_n}}{|\mc{Z}_n|})} \\
\leq &1,
\end{split}
\end{equation}
where $ t^n(T_n^s):=\sum\limits_{x^n \in T_n^s}t^n(x^n)$.

For $\delta>0$, we denote the set of typical types that have maximum deviation $\frac{\delta}{|\mc{R}|}$ from the distribution $t$ as $S_{typical}^{t,\delta}$, i.e., 
\[
S_{typical}^{t,\delta}:=\{s~|~s\in \mc{T}_n^{\mc{R}},~\forall x \in \mc{R},~|s(x)-t(x)|\leq \frac{\delta}{|\mc{R}|}~\text{if}~t(x)>0~\text{else}~s(x)=0\}.
\]
Then, from Eq.~(\ref{equ:l4}), we have
\begin{equation}
\label{equ:l5}
\begin{split}
&1 \\
\geq &\lim_{n \rightarrow \infty}\sup_{s \in S_{typical}^{t,\delta}}  2^{-\frac{1}{|T_n^s|}\sum\limits_{x^n \in  T_n^s} D(\Phi_{f_n}(\tau_{A^n}^{x^n})\| \frac{I_{Z_n}}{|\mc{Z}_n|})}\\
\geq &\lim_{n \rightarrow \infty}\sum\limits_{s \in S_{typical}^{t,\delta}} t^n(T_n^s) 2^{-\frac{1}{|T_n^s|}\sum\limits_{x^n \in  T_n^s} D(\Phi_{f_n}(\tau_{A^n}^{x^n})\| \frac{I_{Z_n}}{|\mc{Z}_n|})} \\
\geq &\lim_{n \rightarrow \infty}\big(\sum\limits_{s \in \mc{T}_n^{\mc{R}}} t^n(T_n^s) 2^{-\frac{1}{|T_n^s|}\sum\limits_{x^n \in  T_n^s} D(\Phi_{f_n}(\tau_{A^n}^{x^n})\| \frac{I_{Z_n}}{|\mc{Z}_n|})} -\sum_{s \notin S_{typical}^{t,\delta}}t^n(T_n^s) \big) \\
=&1,
\end{split}
\end{equation}
where the last line is because that $\lim\limits_{n \rightarrow \infty}\sum\limits_{s \notin S_{typical}^{t,\delta}}t^n(T_n^s)=0$.

Next,  we  evaluate $F(\Phi_{f_n}(\rho_{RA}^{\ox n}), \rho_R^{\ox n} \ox \frac{I_{Z_n}}{|\mc{Z}_n|})$ as 
\begin{equation}
\label{equ:evp1}
\begin{split}
F(\Phi_{f_n}(\rho_{RA}^{\ox n}), \rho_R^{\ox n} \ox \frac{I_{Z_n}}{|\mc{Z}_n|}) 
=&\sum_{x^n \in  \mc{R}^{\times n}}p^n(x^n) F(\Phi_{f_n}(\rho_{A^n}^{x^n}), \frac{I_{Z_n}}{|\mc{Z}_n|}) \\
\geq & \sum_{x^n \in  \mc{R}^{\times n}}p^n(x^n) 2^{-\frac{1}{2}\{D(\Phi_{f_n}(\tau_{A^n}^{x^n})\|\Phi_{f_n}(\rho_{A^n}^{x^n}))+D(\Phi_{f_n}(\tau_{A^n}^{x^n})\| \frac{I_{Z_n}}{|\mc{Z}_n|})\}}\\
\geq & \sum_{x^n \in  \mc{R}^{\times n}}p^n(x^n) 2^{-\frac{1}{2}\{D(\tau_{A^n}^{x^n}\|\rho_{A^n}^{x^n})+D(\Phi_{f_n}(\tau_{A^n}^{x^n})\| \frac{I_{Z_n}}{|\mc{Z}_n|})\}} \\
\geq &\sum_{s \in S_{typical}^{t,\delta}}\sum_{x^n \in T_n^s}p^n(x^n) 2^{-\frac{1}{2}\{D(\tau_{A^n}^{x^n}\|\rho_{A^n}^{x^n})+D(\Phi_{f_n}(\tau_{A^n}^{x^n})\| \frac{I_{Z_n}}{|\mc{Z}_n|})\}} \\
= &\sum_{s \in S_{typical}^{t,\delta}} p^n(T_n^s) 2^{-\frac{1}{2}n\mathbb{E}_{x \sim s}D(\tau_{A}^{x}\|\rho_{A}^{x})} (\frac{1}{|T_n^s|}\sum_{x^n \in T_n^s}2^{-\frac{1}{2}D(\Phi_{f_n}(\tau_{A^n}^{x^n})\| \frac{I_{Z_n}}{|\mc{Z}_n|})})\\
\geq &\sum_{s \in S_{typical}^{t,\delta}} p^n(T_n^s) 2^{-\frac{1}{2}n\mathbb{E}_{x \sim s}D(\tau_{A}^{x}\|\rho_{A}^{x})} 2^{-\frac{1}{2|T_n^s|}\sum\limits_{x^n \in T_n^s} D(\Phi_{f_n}(\tau_{A^n}^{x^n})\| \frac{I_{Z_n}}{|\mc{Z}_n|})},
\end{split}
\end{equation}
where $p^n(T_n^s):=\sum\limits_{x^n \in T_n^s}p^n(x^n)$, the first inequality comes from Lemma~\ref{lem:fidelity-re} in Appendix, the second inequality is due to the data processing inequality of relative entropy and the last inequality is because that the exponential function $2^x$ is convex in $x$. 

We let $p_{\rm{min}}$ be the smallest element in the probability distribution $\{p(x)\}_{x \in \mc{R}}$. Then 
$p^n(T_n^s)$ and $\mathbb{E}_{x \sim s}D(\tau_A^x\|\rho_A^x)$ in the last line of Eq.~(\ref{equ:evp1}) can be bounded as 
\begin{equation}
\label{equ:epv2}
\begin{split}
 &p^n(T_n^s) \\
\geq &(n+1)^{-|R|}2^{-nD(s\|p)} \\
=&(n+1)^{-|R|}2^{-nD(s\|p)+nD(t\|p)} 2^{-nD(t\|p)} \\
=&(n+1)^{-|R|}2^{-n(H(t)-H(s)+\tr t\log p-\tr s \log p)} 2^{-nD(t\|p)} \\
\geq&(n+1)^{-|R|}2^{-n(\delta\log|R|+h(\delta)-2\delta \log p_{\rm{min}})}2^{-nD(t\|p)},
\end{split}
\end{equation}
where the first inequality is due to Eq.~(\ref{eq:proset}) and the last inequality follows from Fannes-Audenaert inequality and H\"{o}lder inequality,  and
\begin{equation}
\label{equ:evp3}
\begin{split}
\mathbb{E}_{x \sim s}D(\tau_A^x\|\rho_A^x)&=\mathbb{E}_{x\sim t} D(\tau_A^x\|\rho_A^x)+\mathbb{E}_{x\sim s}D(\tau_A^x\|\rho_A^x)-\mathbb{E}_{x\sim t}D(\tau_A^x\|\rho_A^x) \\
&\leq \mathbb{E}_{x \sim t}D(\tau_A^x\|\rho_A^x)+2\delta \sup_{x\in \mc{R}} D(\tau_A^x\|\rho_A^x).
\end{split}
\end{equation}
Eq.~(\ref{equ:evp1}), Eq.~(\ref{equ:epv2}) and Eq.~(\ref{equ:evp3}) give
\begin{equation}
\label{equ:finalev}
\begin{split}
&F(\Phi_{f_n}(\rho_{RA}^{\ox n}), \rho_R^{\ox n} \ox \frac{I_{Z_n}}{|\mc{Z}_n|})\\
\geq &(n+1)^{-|R|}2^{-n(\delta\log|R|+h(\delta)-2\delta \log p_{min}+\delta\sup\limits_{x \in \mc{R}}D(\tau_A^x\|\rho_A^x))}2^{-nD(t\|p)-\frac{1}{2}n\mathbb{E}_{x \sim t}D(\tau_A^x\|\rho_A^x)} \\
&\times\sup_{s \in S_{typical}^{t,\delta}} 2^{-\frac{1}{2|T_n^s|}\sum\limits_{x^n \in T_n^s} D(\Phi_{f_n}(\tau_{A^n}^{x^n})\| \frac{I_{Z_n}}{|\mc{Z}_n|})}.
\end{split}
\end{equation}
From Eq.~(\ref{equ:l5}) and Eq.~(\ref{equ:finalev}) and the definition of $E^{\rm{pa}}_{sc}(\rho_{RA},r)$, we can obtain
\begin{equation}
\label{equ:l9}
\begin{split}
&E^{\rm{pa}}_{sc}(\rho_{RA},r) \\
\leq &\lim_{n \rightarrow \infty} \frac{-1}{n} \log F^2(\Phi_{f_n}(\rho_{RA}^{\ox n}), \rho_R^{\ox n} \ox \frac{I_{Z_n}}{|\mc{Z}_n|}) \\
\leq &2D(t\|p)+\mathbb{E}_{x \sim t}D(\tau_A^x\|\rho_A^x)+2(\delta\log|R|+h(\delta)-2\delta \log p_{min}+\delta\sup_{x \in \mc{R}}D(\tau_A^x\|\rho_A^x)) \\
=&2D(t\|p)+\mathbb{E}_{x \sim t} D(\tau_A^x\|\rho_A^x)+|r+\mathbb{E}_tD(\tau_A^x\|I_A)|^+ \\
&+2(\delta\log|R|+h(\delta)-2\delta \log p_{min}+\delta\sup_{x \in \mc{R}}D(\tau_A^x\|\rho_A^x)).
\end{split}
\end{equation}
Because Eq.~(\ref{equ:l9}) holds for any $\delta>0$, by letting $\delta \rightarrow 0$, we complete the proof of 
\textbf{Case 1}.

\textbf{Case 2:}  $r\geq -\mathbb{E}_{x\sim t}D(\tau_A^x\|I_A)$;

For $\tau_{RA}=\sum\limits_{x \in \mc{R}}t(x) \proj{x}_R \ox \tau_A^x$, we let $r'=H(A|R)_{\tau}-\delta'$, where $\delta'>0$ is a constant. From the proof of \textbf{Case 1}, we know that there exists a sequence of hash functions
$\{f'_n~: \mc{A}^{\times n} \rightarrow \mc{Z'}_n=\{1,\ldots,2^{nr'}\}\}_{n \in \mathbb{N}}$ such that 
\begin{equation}
\label{equ:ot1}
\lim_{n \rightarrow \infty} \frac{-1}{n} \log F^2(\Phi_{f'_n}(\rho_{RA}^{\otimes n }),\rho_{R}^{\otimes n}\otimes\frac{I_{Z'_n}}{|\mathcal{Z'}_n|})
\leq 2D(t\|p)+\mathbb{E}_{x\sim t} D(\tau_A^x\|\rho_A^x).
\end{equation}
Then, We convert $\{f'_n\}_{n\in \mathbb{N}}$ to a sequence of
hash functions $\{f_n~: \mc{A}^{\times n} \rightarrow \mc{Z}_n=\{1,\ldots,2^{nr}\}\}_{n \in \mathbb{N}}$, by extending the sizes of the extracted randomness but otherwise letting $f_n=f'_n$. Then, we have
\begin{equation}
\label{equ:ot2}
\begin{split}
&F(\Phi_{f_n}(\rho_{RA}^{\otimes n }),\rho_{R}^{\otimes n}\otimes\frac{I_{Z_n}}{|\mathcal{Z}_n|}) \\
=& F(\Phi_{f'_n}(\rho_{RA}^{\otimes n }),\rho_{R}^{\otimes n}\otimes\frac{I_{Z'_n}}{|\mathcal{Z}_n|}) \\
=&\sqrt{\frac{|\mathcal{Z'}_n|}{|\mathcal{Z}_n|}}F(\Phi_{f'_n}(\rho_{RA}^{\otimes n }),\rho_{R}^{\otimes n}\otimes\frac{I_{Z'_n}}{|\mathcal{Z'}_n|})
\end{split}
\end{equation}
Eq.~(\ref{equ:ot1}) and Eq.~(\ref{equ:ot2}) imply that
\begin{equation}
\label{equ:ot3}
\begin{split}
&E^{\rm{pa}}_{sc}(\rho_{RA},r) \\
\leq &\lim_{n \rightarrow \infty} \frac{-1}{n} \log F^2(\Phi_{f_n}(\rho_{RA}^{\otimes n }),\rho_{R}^{\otimes n}\otimes\frac{I_{Z_n}}{|\mathcal{Z}_n|}) \\
\leq &2D(t\|p)+\mathbb{E}_{x \sim t} D(\tau_A^x\|\rho_A^x)+r-H(A|R)_\tau+\delta'\\
=&2D(t\|p)+\mathbb{E}_{x \sim t} D(\tau_A^x\|\rho_A^x)+|r-H(A|R)_\tau|^+ +\delta'
\end{split}
\end{equation}
Noticing that Eq.~(\ref{equ:ot3}) holds for any $\delta'>0$, let $\delta' \rightarrow 0$ and
we are done.
\end{proofof}


\subsection{Achievability}

In this subsection, we use the strong converse exponent of classical privacy amplification to derive the achievability part of Eq.~(\ref{equ:claconditional}).

\begin{proof}
If $r<0$, it is obvious that all terms in Eq.~(\ref{equ:claconditional}) equal to $0$ and the assertion holds trivially. Hence, for the rest we assume that $r \geq 0$.

For any function $f_n~:~\mc{A}^{\times n}\rightarrow \mc{Z}_n=\{1,\ldots,2^{nr}\}$, we can obtain from Lemma~\ref{lem:monohash} in Appendix that
\begin{equation}
\label{equ:hato}
 H^{\epsilon,P}_{\rm{min}}(A^n|\dot{R}^n)_{\rho^{\ox n}} \geq  H^{\epsilon,P}_{\rm{min}}(Z_n|\dot{R}^n)_{\Phi_{f_n}(\rho^{\ox n})}.
\end{equation}
Eq.~(\ref{equ:hato}) implies that 
\begin{equation}
\label{equ:hotore}
 \epsilon^P_{Z^n|\dot{R}^n}(\Phi_{f_n}(\rho_{RA}^{\ox n}), nr) \geq \epsilon^P_{A^n|\dot{R}^n}(\rho_{RA}^{\ox n},  nr)
\end{equation}
It is easy to see that
\begin{equation}
\label{equ:hato1}
P(\Phi_{f_n}(\rho_{RA}^{\ox n}), \rho_R^{\ox n} \ox \frac{I_{Z_n}}{|\mc{Z}_n|}) \geq \epsilon^P_{Z^n|\dot{R}^n}(\Phi_{f_n}(\rho_{RA}^{\ox n}), nr).
\end{equation}
Eq.~(\ref{equ:hotore}) and Eq.~(\ref{equ:hato1}) give 
\begin{equation}
\label{equ:hatofianl}
\inf_{f_n \in \mc{F}_n(r)}P(\Phi_{f_n}(\rho_{RA}^{\ox n}, \rho_R^{\ox n} \ox \frac{I_{Z_n}}{|\mc{Z}_n|})
\geq \epsilon^P_{A^n|\dot{R}^n}(\rho_{RA}^{\ox n},  nr).
\end{equation}
The achievability part follows from Theorem~\ref{thm:privacy} together with Eq.~(\ref{equ:hatofianl}).
\end{proof}


\subsection{Optimality}

In this subsection, we prove the optimality part of Eq.~(\ref{equ:claconditional}).

\begin{proof}
Noticing that $\rho_{RA}$ is classical, the minimization in $\epsilon^P_{A^n|\dot{R}^n}(\rho_{RA}^{\ox n},nr)$ can be restricted to over the classical sub-normalized states. Then for any classical sub-normalized state $\tilde{\rho}_{R^nA^n}=\sum\limits_{x^n \in \mc{R}^{\times n}}\proj{x^n}_{R^n} \ox \tilde{\rho}^{x^n}_{A^n}$ which satisfies
\begin{equation}
\label{equ:rescon}
    \tilde{\rho}_{R^nA^n} \leq 2^{-nr} I_{A}^{\ox n} \ox \rho_{R}^{\ox n}, ~\tilde{\rho}_{R^n} \leq \rho^{\ox n}_{R},
\end{equation}
we can evaluate $F( \rho_{RA}^{\ox n},  \tilde{\rho}_{R^nA^n})$ as

\begin{equation}
\label{equ:oprev}
\begin{split}
F(\rho_{RA}^{\ox n}, \tilde{\rho}_{R^nA^n})&=\sum_{x^n \in \mc{R}^{\times n}} p^n(x^n)F(\rho_{A^n}^{x^n}, \frac{\tilde{\rho}^{x^n}_{A^n}}{p^n(x^n)}) \\
&\leq \sum_{x^n \in \mc{R}^{\times n}} p^n(x^n) 2^{\frac{1-\alpha}{2\alpha}(-D_\alpha(\rho_{A^n}^{x^n}\|I_A^{\ox n})+D_\beta(\frac{\tilde{\rho}^{x^n}_{A^n}}{p^n(x^n)}\| I_A^{\ox n})+\frac{1}{\beta-1}\log \frac{\tr \tilde{\rho}^{x^n}_{A^n}}{p^n(x^n)})} \\
&\leq \sum_{x^n \in \mc{R}^{\times n}} p^n(x^n) 2^{\frac{1-\alpha}{2\alpha}(-D_\alpha(\rho_{A^n}^{x^n}\|I_A^{\ox n})-nr)} \\
&=\sum_{t \in \mc{T}_n^{\mc{R}}} (\sum_{x^n \in T_n^t}p^n(x^n))2^{\frac{1-\alpha}{2\alpha}(-n\mathbb{E}_{x \sim t} D_\alpha(\rho_A^x\|I_A)-nr)} \\
&\leq (n+1)^{|R|} \sup_{t \in \mc{T}_n^{\mc{R}}}2^{-nD(t\|p)}2^{\frac{1-\alpha}{2\alpha}(-n\mathbb{E}_{x\sim t} D_\alpha(\rho_A^x\|I_A)-nr)},
\end{split}
\end{equation}
where the first inequality follows from Lemma~\ref{lem:hof} in Appendix, the second inequality is because that by Eq.~(\ref{equ:rescon}), $\frac{\tilde{\rho}^{x^n}_{A^n}}{p^n(x^n)} \leq 2^{-nr} I_A^{\ox n}$, $\tr \tilde{\rho}^{x^n}_{A^n} \leq p^n(x^n)$, $\forall~x^n \in \mc{R}^{\times n}$ and the last inequality comes from Eq.~(\ref{eq:typenumber}) and Eq.~(\ref{eq:proset}).

Eq.~(\ref{equ:oprev}) gives 
\begin{equation}
\label{equ:conoptre}
\begin{split}
&\lim_{n\rightarrow \infty} \frac{-1}{n} \log (1-\epsilon^P_{A^n|\dot{R}^n}(\rho^{\ox n}_{RA},nr)) \\
\geq &\lim_{n\rightarrow \infty} \big(\inf_{t \in \mc{T}_n^{\mc{R}}} \{2D(t\|p)+\frac{1-\alpha}{\alpha}(\mathbb{E}_{x \sim t} D_\alpha(\rho_A^x\|I_A)+r)\}-2|R|\frac{\log (n+1)}{n} \big)\\
=& \inf_{t \in \mc{Q}({\mc{R}})} \{2D(t\|p)+\frac{1-\alpha}{\alpha}(\mathbb{E}_{x \sim t}D_\alpha(\rho_A^x\|I_A)+r)\}\\
=& \inf_{t \in \mc{Q}({\mc{R}})}\big\{2D(t\|p)+\frac{1-\alpha}{\alpha}\big(r-\mathbb{E}_{x\sim t}H_\alpha(p(\cdot|x))\big)\big\}.
\end{split}
\end{equation}
Because Eq.~(\ref{equ:conoptre}) holds for any $\frac{1}{2}<\alpha <1$, we complete the proof.
\end{proof}


\section{Partially smoothed conditional min-entropy for pure states}
\label{sec:constong}

In this section, we establish the strong converse exponent for  partially smoothing of the conditional min-entropy for pure states, i.e., Eq.~(\ref{equ:main}). We divide the process into several steps. At first, we give a sequence of quantities with the same exponential decay rate as $\{1-\epsilon_{A^n|\dot{R^n}}(\rho_{RA}^{\ox n},nr)\}_{n \in \mathbb{N}}$.

\begin{proposition}
\label{prop:pin}
For pure state $\rho_{RA}$  stated in Theorem~\ref{thm:main}, $r \in \mathbb{R}$
and $n \in \mathbb{N}$, we define $\Pi_t:=\sum\limits_{x^n \in T_n^t} \proj{x^n}_{R^n}$ and  $P(t):=\sum\limits_{x^n \in T_n^t}  p^n(x^n)$ for any $t \in \mc{T}_n^\mc{X}$. Then, we have
\begin{equation}
\label{equ:coneequi1}
\begin{split}
&\frac{1}{\sqrt{(n+1)^{|R|}}}\sum_{t \in \mc{T}_n^\mc{X}}P(t) \sup\{F(\sigma_{R^nA^n}^t,\frac{\Pi_t\rho_{RA}^{\ox n} \Pi_t}{P(t)})~|~\sigma_{R^nA^n}^t \in \mc{F}_t \}  \\
\leq &\sup\{F(\sigma_{R^nA^n},\rho_{RA}^{\ox n})~|~\sigma_{R^nA^n}\in \mc{S}_{\leq}(R^nA^n):\sigma_{R^nA^n}\leq 2^{-nr}I_{A}^{\ox n}\ox \rho_{R}^{\ox n}, \sigma_{R^n} \leq \rho_{R}^{\ox n}\} \\
\leq & \sum_{t \in \mc{T}_n^\mc{X}}P(t) \sup\{F(\sigma_{R^nA^n}^t,\frac{\Pi_t\rho_{RA}^{\ox n} \Pi_t}{P(t)})~|~\sigma_{R^nA^n}^t \in \mc{F}_t \}  ,
\end{split}
\end{equation}
where $\mc{F}_t:=\{\sigma_{R^nA^n}^t~|~ \sigma_{R^nA^n}^t \in \mc{S}_{\leq}(R^nA^n), \sigma_{R^nA^n}^t \leq 2^{-nr}I_{A}^{\ox n}\ox \frac{\Pi_t}{|T_n^t|}, \sigma_{R^n}^t \leq \frac{\Pi_t}{|T_n^t|} \}$, $\forall t \in \mc{T}_n^\mc{X}$.
\end{proposition}

\begin{proof}
We denote by $\mc{E}_n$  the pinching channel $\mc{E}_n(X)=\sum\limits_{t \in \mc{T}_n^{\mc{X}}} \Pi_t X \Pi_t$. 
On the one hand, for any $\sigma_{R^nA^n}$ which satisfies the restriction in Eq.~(\ref{equ:coneequi1}) and $t \in \mc{T}_n^\mc{X}$, we have
\begin{equation}
\label{equ:conr}
\begin{split}
\frac{\Pi_t\sigma_{R^nA^n}\Pi_t}{P(t)}&\leq 2^{-nr}I_{A}^{\ox n}\ox \frac{\Pi_t}{|T_n^t|}, \\
\frac{\Pi_t\sigma_{R^n}\Pi_t}{P(t)} &\leq \frac{\Pi_t}{|T_n^t|}.
\end{split}
\end{equation}
On the other hand, we can get from the data processing inequality of the fidelity function that
\begin{equation}
\label{equ:sec}
\begin{split}
&F(\sigma_{R^nA^n},\rho_{RA}^{\ox n}) \\
\leq &F(\mc{E}_n(\sigma_{R^nA^n}),\mc{E}_n(\rho_{RA}^{\ox n}))\\
=&\sum_{t \in \mc{T}_n^\mc{X}} P(t)F(\frac{\Pi_t\sigma_{R^nA^n}\Pi_t}{P(t)},\frac{\Pi_t\rho_{RA}^{\ox n} \Pi_t}{P(t)}).
\end{split}
\end{equation}
Eq.~(\ref{equ:conr}) and Eq.~(\ref{equ:sec}) give the second inequality in Eq.~(\ref{equ:coneequi1}). For the first inequality in Eq.~(\ref{equ:coneequi1}), noticing that for any $\{\sigma_{R^nA^n}^t~|~\sigma_{R^nA^n}^t \in \mc{F}_t\}_{t \in \mc{T}_n^\mc{X}}$, we have
\begin{equation}
\label{equ:one}
\begin{split}
&\sum_{t \in \mc{T}_n^\mc{X}}P(t)F(\sigma_{R^nA^n}^t,\frac{\Pi_t\rho_{RA}^{\ox n} \Pi_t}{P(t)})\\
=&F(\sum_{t \in \mc{T}_n^\mc{X}}P(t)\sigma_{R^nA^n}^t, \sum_{t \in \mc{T}_n^\mc{X}}\Pi_t\rho_{RA}^{\ox n} \Pi_t) \\
\leq &\sqrt{(n+1)^{|R|}} F(\sum_{t \in \mc{T}_n^\mc{X}}P(t)\sigma_{R^nA^n}^t, \rho_{RA}^{\ox n}),
\end{split}
\end{equation}
where the inequality is due to Lemma~\ref{lem:appen1} in Appendix. Because $\sum\limits_{t \in \mc{T}_n^\mc{X}}P(t)\sigma_{A^nR^n}^t \leq 2^{-nr}I_{A}^{\ox n}\ox \rho_{R}^{\ox n}$ and $\sum\limits_{t \in \mc{T}_n^\mc{X}}P(t)\sigma_{R^n}^t \leq \rho_{R}^{\ox n}$, Eq.~(\ref{equ:one}) implies that the first inequality in Eq.~(\ref{equ:coneequi1}) holds.
\end{proof}

Next, we establish an equivalent expression for the sequence of quantities
\begin{align}
\left\{\sup\left\{F\left(\sigma_{R^nA^n}^t,\frac{\Pi_t\rho_{RA}^{\ox n} \Pi_t}{P(t)}\right)~\middle|~\sigma_{R^nA^n}^t \in \mc{F}_t \right\} \right\}_{t \in \mc{T}_n^\mc{X}}.
\end{align}

\begin{proposition}
\label{prop:equi}
For $\rho_{RA}$, $r$, $\{\Pi_t\}_{t \in \mc{T}_n^\mc{X}}$, $\{P(t)\}_{t \in \mc{T}_n^\mc{X}}$ and $\{\mc{F}_t\}_{t \in \mc{T}_n^\mc{X}}$ defined in Proposition~\ref{prop:pin}, we have
\begin{equation}
\label{equ:equcon}
\sup\{F(\sigma_{R^nA^n}^t,\frac{\Pi_t\rho_{RA}^{\ox n} \Pi_t}{P(t)})~|~\sigma_{R^nA^n}^t \in \mc{F}_t \}=\sup\{F(\tau_{R^n}^t,\psi_t)~|~ \tau_{R^n}^t \in \mc{G}_t  \},~\forall t \in \mc{T}_n^\mc{X},
\end{equation}
where $\mc{G}_t:=\{\tau_{R^n}^t~|~\tau_{R^n}^t \in \mc{S}_{\leq}(R^n), \tau_{R^n}^t \leq 2^{-nr}\frac{\Pi_t}{|T_n^t|}, \sum\limits_{x^n \in T^t_n} \proj{x^n}\tau_{R^n}^t \proj{x^n} \leq \frac{\Pi_t}{|T_n^t|}\}$ and
$\ket{\psi_t}:=\sum\limits_{x^n \in T_n^t} \frac{\ket{x^n}_{R^n}}{\sqrt{|T_n^t|}}$.
\end{proposition}
\begin{proof}
We let $Q_t:=\sum\limits_{x^n \in T_n^t} \proj{x^n}_{R^n}\ox\proj{x^n}_{A^n}$. For any $\sigma_{R^nA^n}^t \in \mc{F}_t $, we have
\begin{equation}
\label{equ:2con}
\begin{split}
&Q_t \sigma_{R^nA^n}^t Q_t=\sum_{x^n,x'^n \in T_n^t} a_{x^n,x'^n} \ket{x^n}\bra{x'^n}_{R^n}\ox \ket{x^n}\bra{x'^n}_{A^n} \leq 2^{-nr}\frac{1}{|T_n^t|} \sum_{x^n \in T_n^t}\proj{x^n}_{R^n} \ox \proj{x^n}_{A^n},\\
&\tr_{A^n}(Q_t \sigma_{R^nA^n}^t Q_t)=\sum_{x^n \in T_n^t} a_{x^n,x^n} \proj{x^n}_{R^n} \leq \sum_{x^n \in T_n^t} \proj{x^n}_{R^n} \sigma_{R^n}^t \proj{x^n}_{R^n}\leq \frac{\Pi_t}{|T_n^t|},
\end{split}
\end{equation}
where $a_{x^n,x'^n}=\bra{x^n}_{A^n}\ox \bra{x^n}_{R^n} \sigma_{A^nR^n}^t \ket{x'^n}_{A^n}\ox \ket{x'^n}_{R^n} $. We denote the operator $\sum\limits_{x^n,x'^n \in T_n^t} a_{x^n,x'^n}\ket{x^n}\bra{x'^n}_{R^n}$ and $\sum\limits_{x^n \in T_n^t} \proj{x^n}_{R^n} \ox \ket{x^n}_{A^n}$ as $\tau_{R^n}^t$ and $V_t$, respectively. We can easily get from Eq.~(\ref{equ:2con}) that
\begin{equation}
\label{equ:equi0}
\begin{split}
&\tau^t_{R^n}=V_t^* Q_t \sigma_{R^nA^n}^t Q_t V_t \leq 2^{-nr} \frac{\Pi_t}{|T_n^t|} \\
&\sum_{x^n \in T_n^t} \proj{x^n}_{R^n} \tau_{R^n}^t \proj{x^n}_{R^n}=a_{x^n,x^n} \proj{x^n}_{R^n} \leq \frac{\Pi_t}{|T_n^t|}.
\end{split}
\end{equation}
Hence $\tau_{R^n}^t \in \mc{G}_t$ and 
\begin{equation}
\label{equ:equi1}
\begin{split}
&F(\tau^t_{R^n},\psi_t) \\
=&F(V_t(\tau^t_{R^n})V_t^*,V_t(\psi_t)V_t^*) \\
=&F(Q_t \sigma_{R^nA^n}^t Q_t,\frac{\Pi_t\rho_{RA}^{\ox n} \Pi_t}{P(t)})\\
=&F( \sigma_{R^nA^n}^t ,\frac{\Pi_t\rho_{RA}^{\ox n} \Pi_t}{P(t)}).
\end{split}
\end{equation}
In turn, for any $\tau_{R^n}^t=\sum\limits_{x^n, x'^n \in T_n^t} b_{x^n,x'^n}\ket{x^n}\bra{x'^n}_{R^n} \in \mc{G}_t$, we have
\begin{equation}
\label{equ:equi2}
\begin{split}
&V_t \tau_{R^n}^t V_t^* \leq 2^{-nr} V_t \frac{\Pi_t}{|T_n^t|} V_t^* \leq 2^{-nr}I_{A}^{\ox n}\ox \frac{\Pi_t}{|T_n^t|}, \\
&\tr_{A^n}(V_t \tau_{R^n}^t V_t^*)=\sum_{x^n \in T_n^t} \proj{x^n}_{R^n} \tau_{R^n}^t \proj{x^n}_{R^n} \leq \frac{\Pi_t}{|T_n^t|},\\
&F(V_t \tau_{R^n}^t V_t^*,\frac{\Pi_t\rho_{RA}^{\ox n} \Pi_t}{P(t)})=F(V_t \tau_{R^n}^t V_t^*,V_t(\psi_t)V_t^*)=F(\tau_{R^n}^t ,\psi_t ).
\end{split}
\end{equation}
Eq.~(\ref{equ:equi0}), Eq.~(\ref{equ:equi1}) together with Eq.~(\ref{equ:equi2}) imply that Eq.~(\ref{equ:equcon}) holds. 
\end{proof}

With the above results in hand, we can start the proof of Eq.~(\ref{equ:main}) now. 

\begin{proofof}[of Eq.~(\ref{equ:main})]
First of all, we prove ``$\geq$'' part. From Proposition~\ref{prop:pin} and Proposition~\ref{prop:equi}, we have
\begin{equation}
\label{equ:convcon}
\begin{split}
&\sup\{F(\sigma_{R^nA^n},\rho_{RA}^{\ox n})~|~\sigma_{R^nA^n}\in \mc{S}_{\leq}(R^nA^n):\sigma_{R^nA^n}\leq 2^{-nr}I_{A}^{\ox n}\ox \rho_{R}^{\ox n}, \sigma_{R^n} \leq \rho_{R}^{\ox n}\} \\
\leq & \sum_{t \in \mc{T}_n^\mc{X}}P(t) \sup\{F(\tau_{A^n}^t,\psi_t)~|~ \tau_{A^n}^t \in \mc{G}_t  \}  \\
\leq &(n+1)^{|R|} \sup_{t \in \mc{T}_n^\mc{X}}2^{-nD(t\|p)} \sup_{\tau_{A^n}^t:\tau_{A^n}^t \in \mc{G}_t} F(\tau_{A^n}^t,\psi_t) \\
\leq &(n+1)^{|R|} \sup_{t \in \mc{T}_n^\mc{X}} 2^{-nD(t\|p)} \sup_{\tau_{A^n}^t:\tau_{A^n}^t \in \mc{G}_t} 2^{\frac{1-\alpha}{2\alpha}(-D^*_{\alpha}(\psi_t \| \frac{\Pi_t}{|T_n^t|})+D^*_\beta(\tau_{A^n}^t\|\frac{\Pi_t}{|T_n^t|})+\frac{1}{\beta-1}\log\tr \tau_{A^n}^t)} \\
\leq &(n+1)^{|R|} \sup_{t \in \mc{T}_n^\mc{X}} 2^{-nD(t\|p)}  2^{\frac{1-\alpha}{2\alpha}(-D^*_{\alpha}(\psi_t \| \frac{\Pi_t}{|T_n^t|})-nr)} \\
=&(n+1)^{|R|} \sup_{t \in \mc{T}_n^\mc{X}} 2^{-nD(t\|p)}  2^{\frac{1-\alpha}{2\alpha}(-\log|T_n^t|-nr)} \\
\leq &(n+1)^{\frac{3|R|}{2}} \sup_{t \in \mc{T}_n^\mc{X}} 2^{-nD(t\|p)}  2^{\frac{1-\alpha}{2\alpha}(-H(t)-nr)},
\end{split}
\end{equation}
where the second inequality comes from Eq.~(\ref{eq:typenumber}) and Eq.~(\ref{eq:proset}), the third inequality is from Lemma~\ref{lem:hof} in Appendix, the fourth inequality follows from Proposition~\ref{prop:mainpro}~(\romannumeral1) and $\tr \tau_{A^n}^t \leq 1,~\forall~\tau_{A^n}^t \in \mc{G}_t$ and the last inequality is due to Eq.~(\ref{eq:numt}). 

Eq.~(\ref{equ:convcon}) implies that
\begin{equation}
\label{equ:licon}
\begin{split}
&\lim_{n \rightarrow \infty}\frac{-1}{n} \log \sup\{F^2(\sigma_{R^nA^n},\rho_{RA}^{\ox n})~|~\sigma_{R^nA^n}\in \mc{S}_{\leq}(R^nA^n):\sigma_{R^nA^n}\leq 2^{-nr}I_{A}^{\ox n}\ox \rho_{R}^{\ox n}, \sigma_{R^n} \leq \rho_{R}^{\ox n}\} \\
\geq &\lim_{n \rightarrow \infty} (\inf_{t \in \mc{T}_n^\mc{X}} \big\{2D(t\|p)+\frac{1-\alpha}{\alpha}(H(t)+nr)  \big\}-3|R| \frac{\log(n+1)}{n})\\
=&\inf_{t \in \mc{Q}(\mc{X})} \big\{2D(t\|p)+\frac{1-\alpha}{\alpha}(H(t)+nr)\big\}.
\end{split}
\end{equation}
Because Eq.~(\ref{equ:licon}) holds for any $\frac{1}{2}<\alpha<1$, we have
\begin{equation}
\begin{split}
&\lim_{n \rightarrow \infty}\frac{-1}{n} \big(1-\epsilon^P_{A^n|\dot{R^n}}(\rho_{RA}^{\ox n},nr)\big) \\
\geq &\lim_{n \rightarrow \infty}\frac{-1}{n} \log \sup\{F^2(\sigma_{R^nA^n},\rho_{RA}^{\ox n})~|~\sigma_{R^nA^n}\in \mc{S}_{\leq}(R^nA^n):\sigma_{R^nA^n}\leq 2^{-nr}I_{A}^{\ox n}\ox \rho_{R}^{\ox n}, \sigma_{R^n} \leq \rho_{R}^{\ox n}\} \\
\geq &\sup_{\frac{1}{2}<\alpha<1} \inf_{t \in \mc{Q}(\mc{X})} \big\{2D(t\|p)+\frac{1-\alpha}{\alpha}(H(t)+nr)\big\} \\
=&\sup_{0<s<1} \inf_{t \in \mc{Q}(\mc{X})} \big\{2D(t\|p)+s(H(t)+nr) \big\} \\
=&\inf_{t \in \mc{Q}(\mc{X})} \sup_{0<s<1} \big\{2D(t\|p)+s(H(t)+nr) \big\} \\
=&\inf_{t \in \mc{Q}(\mc{X})} \big\{2D(t\|p)+|H(t)+nr|^+ \big\}.
\end{split}
\end{equation}

Next, we turn to the proof of ``$\leq$'' part. For $t \in \mc{T}_n^\mc{X}$ such that $\frac{2^{-nr}}{|T_n^t|} \geq 1$, we have $\psi_t \in \mc{G}_t$. Hence,
\begin{equation}
\label{equ:lowbo1}
\sup\{F(\sigma_{A^n}^t,\psi_t)~|~ \sigma_{A^n}^t \in \mc{G}_t  \} \geq F(\psi_t,\psi_t)=1=\min\{1,\sqrt{\frac{2^{-nr}}{|T_n^t|}}\}.
\end{equation}
For $t \in \mc{T}_n^\mc{X}$ which satisfies $\frac{2^{-nr}}{|T_n^t|} <1$,  $\frac{2^{-nr}}{|T_n^t|}\psi_t \in \mc{G}_t$. In this case, we also have
\begin{equation}
\label{equ:lowbo2}
\sup\{F(\sigma_{A^n}^t,\psi_t)~|~ \sigma_{A^n}^t \in \mc{G}_t  \} \geq F(\psi_t,\frac{2^{-nr}}{|T_n^t|}\psi_t)=\sqrt{\frac{2^{-nr}}{|T_n^t|}}=\min\{1,\sqrt{\frac{2^{-nr}}{|T_n^t|}}\}.
\end{equation}
Eq.~(\ref{equ:lowbo1}) and Eq.~(\ref{equ:lowbo2}) imply that for any $t \in \mc{T}_n^\mc{X}$, it holds that
\begin{equation}
\label{equ:lowbo}
\sup\{F(\sigma_{A^n}^t,\psi_t)~|~ \sigma_{A^n}^t \in \mc{G}_t  \} \geq \min\{1,\sqrt{\frac{2^{-nr}}{|T_n^t|}}\}.
\end{equation}
From Eq.~(\ref{equ:lowbo}), Proposition~\ref{prop:pin} and Proposition~\ref{prop:equi}, we can obtain
\begin{equation}
\label{equ:achiv}
\begin{split}
&\sup\{F(\sigma_{R^nA^n},\rho_{RA}^{\ox n})~|~\sigma_{R^nA^n}\in \mc{S}_{\leq}(R^nA^n):\sigma_{R^nA^n}\leq 2^{-nr}I_{A}^{\ox n}\ox \rho_{R}^{\ox n}, \sigma_{R^n} \leq \rho_{R}^{\ox n}\} \\
\geq &\frac{1}{\sqrt{(n+1)^{|R|}}}\sup_{t \in \mc{T}_n^\mc{X}} P(t)\min\{1,\sqrt{\frac{2^{-nr}}{|T_n^t|}}\} \\
\geq &(n+1)^{-\frac{3|R|}{2}}\sup_{t \in \mc{T}_n^\mc{X}} 2^{-nD(t\|p)}\min\{1,\sqrt{\frac{2^{-nr}}{2^{nH(t)}}}\},
\end{split}
\end{equation}
where the second inequality is from Eq.~(\ref{eq:numt}) and Eq.~(\ref{eq:proset}). Eq.~(\ref{equ:achiv}) gives 
\begin{equation}
\begin{split}
&\lim_{n \rightarrow \infty}\frac{-1}{n} \big(1-\epsilon^P_{A^n|\dot{R^n}}(\rho_{RA}^{\ox n},nr)\big) \\
\leq &\lim_{n \rightarrow \infty}\frac{-1}{n} \log \frac{1}{2}\sup\{F^2(\sigma_{R^nA^n},\rho_{RA}^{\ox n})~|~\sigma_{R^nA^n}\in \mc{S}_{\leq}(R^nA^n):\sigma_{R^nA^n}\leq 2^{-nr}I_{A}^{\ox n}\ox \rho_{R}^{\ox n}, \sigma_{R^n} \leq \rho_{R}^{\ox n}\}\\
\leq & \lim_{n \rightarrow \infty} (\inf_{t \in \mc{T}_n^\mc{X}} \big\{2D(t\|p)+\max\{0,r+H(t)\}  \big\}+3|R|\frac{\log(n+1)}{n})\\
= &\inf_{t \in \mc{Q}(\mc{X})} \big\{2D(t\|p)+|H(t)+r|^+ \big\}.
\end{split}
\end{equation}
We are done.
\end{proofof}


\section{Partially smoothed mutual max-information for pure states}
\label{sec:strmutualpure}

In this section, we determine the strong converse exponent for partially smoothing of the mutual max-information for pure states, i.e., Eq.~(\ref{equ:puremutual}). We establish the optimality part and the achievability part in the following subsections, respectively.


\subsection{Optimality}

We first establish a lower bound of the strong converse exponent for partially smoothing of the mutual max-information for any quantum state $\rho_{RA}$.
The optimality part will follow from this lower bound directly.

\begin{proposition}
\label{propo:main}
For $\rho_{RA} \in \mc{S}(RA)$, $r \in \mathbb{R}$ and $\Delta \in \{d, P\}$, we have
\begin{equation}
\label{equ:genebound1}
\lim_{n \rightarrow \infty} \frac{-1}{n} \log(1-\epsilon^\Delta_{\dot{R^n}:A^n}(\rho_{RA}^{\otimes n}, nr)) \geq \sup_{\frac{1}{2}<\alpha<1} \frac{1-\alpha}{\alpha}\big\{H_{\beta}(R)_\rho-H^*_{\alpha}(R|A)_\rho-r \big\},
\end{equation}
where $\frac{1}{\alpha}+\frac{1}{\beta}=2$.
\end{proposition}

\begin{proof}
By Fuchs–van de Graaf inequality, we have that for any $n \in \mathbb{N}$,
\begin{equation}
\label{equ:dPre}
\epsilon^d_{\dot{R^n}:A^n}(\rho_{RA}^{\otimes n}, nr) \leq \epsilon^P_{\dot{R^n}:A^n}(\rho_{RA}^{\otimes n}, nr).
\end{equation}
Eq.~(\ref{equ:dPre}) implies that we only need to establish Eq.~(\ref{equ:genebound1}) for the purified distance in the following.

For $\rho_R$ with spectral decomposition $\rho_R=\sum\limits_{x \in \mc{X}} p(x)\proj{x}$, we can write $\rho_R^{\ox n}$ as
\begin{equation}
\rho_{R}^{\ox n}=\sum_{t \in \mathcal{T}_n^{\mc{X}}} P(t) \frac{\Pi_t}{|T_n^t|},
\end{equation}
where  $P(t)=\sum\limits_{x^n \in T_n^t} p^n(x^n)$, $\Pi_t=\sum\limits_{x^n \in T_n^t} \proj{x^n}$.

We let $\mc{E}_n$ be the pinching channel $\mc{E}_n(X):=\sum\limits_{t \in \mc{T}_n^{\mc{X}}} \Pi_t X \Pi_t$. For any $\sigma_{A^n} \in \mc{S}({A^n})$, $\tilde{\rho}_{R^nA^n} \in \mc{S}({R^nA^n})$ satisfy 
\begin{align}
\tilde{\rho}_{R^n}&=\rho_{R}^{\ox n}, \label{1}\\
\tilde{\rho}_{R^nA^n} &\leq 2^{nr} \rho_{R}^{\ox n} \ox \sigma_{A^n} \label{2}
\end{align}
and $\frac{1}{2}<\alpha<1$, $F^2(\tilde{\rho}_{R^nA^n}, \rho_{RA}^{\ox n})$ can be upper bounded as 
\begin{equation}
\label{equ:ev}
\begin{split}
&F^2(\tilde{\rho}_{R^nA^n}, \rho_{RA}^{\ox n}) \\
\leq &F^2(\mc{E}_n(\tilde{\rho}_{R^nA^n}), \mc{E}_n(\rho_{RA}^{\ox n})) \\
=& \big(\sum_{t \in \mathcal{T}_n^{\mc{X}}} F( \Pi_t\tilde{\rho}_{R^nA^n} \Pi_t, \Pi_t \rho_{RA}^{\ox n} \Pi_t)\big)^2\\
=&\big(\sum_{t \in \mathcal{T}_n^{\mc{X}}} P(t) F(\frac{\Pi_t\tilde{\rho}_{R^nA^n} \Pi_t}{P(t)}, \frac{\Pi_t \rho_{RA}^{\ox n} \Pi_t}{P(t)})\big)^2 \\
\leq & \left( \sum_{t \in \mathcal{T}_n^{\mc{X}}} P(t) 2^{\frac{1-\alpha}{2\alpha}\big\{-D^*_{\alpha}( \frac{\Pi_t \rho_{RA}^{\ox n} \Pi_t}{P(t)} \| \frac{\Pi_t}{|T_n^t|} \ox \sigma_{A^n})+D^*_{\beta}(\frac{\Pi_t\tilde{\rho}_{R^nA^n} \Pi_t}{P(t)}\| \frac{\Pi_t}{|T_n^t|} \ox \sigma_{A^n})\big\}} \right)^2,
\end{split}
\end{equation}
where the last inequality comes from Lemma~\ref{lem:hof} in Appendix.

In the following, we evaluate $D^*_\alpha(\cdot\|\cdot)$ and $D^*_{\beta}(\cdot\|\cdot)$ in Eq.~(\ref{equ:ev}), respectively. On the one hand, because $\tilde{\rho}_{R^nA^n}$ satisfies Eq.~(\ref{2}), $\frac{\Pi_t\tilde{\rho}_{R^nA^n} \Pi_t}{P(t)} \leq 2^{nr}\frac{\Pi_t}{|T_n^t|} \ox \sigma_{A^n}$. Furthermore, we get from Proposition~\ref{prop:mainpro} (\romannumeral1) that
\begin{equation}
\label{equ:1ev}
D^*_{\beta}(\frac{\Pi_t\tilde{\rho}_{R^nA^n} \Pi_t}{P(t)}\| \frac{\Pi_t}{|T_n^t|} \ox \sigma_{A^n}) \leq nr.
\end{equation}
On the other hand, 
\begin{equation}
\label{equ:2ev}
\begin{split}
&-D^*_{\alpha}( \frac{\Pi_t \rho_{RA}^{\ox n} \Pi_t}{P(t)} \| \frac{\Pi_t}{|T_n^t|} \ox \sigma_{A^n}) \\
=&-D^*_{\alpha}( \frac{\Pi_t \rho_{RA}^{\ox n} \Pi_t}{P(t)} \| \Pi_t \ox \sigma_{A^n})-\log|T_n^t| \\
\leq &-\inf_{\tau_{A^n} \in \mc{S}(A^n)}D^*_{\alpha}( \frac{\Pi_t \rho_{RA}^{\ox n} \Pi_t}{P(t)} \| \Pi_t \ox \tau_{A^n})-\log|T_n^t| \\
=&-\inf_{\tau_{A^n} \in \mc{S}_{\rm{sym}}(A^n)}D^*_{\alpha}( \frac{\Pi_t \rho_{RA}^{\ox n} \Pi_t}{P(t)} \| \Pi_t \ox \tau_{A^n})-\log|T_n^t| \\
\leq &-D^*_{\alpha}( \frac{\Pi_t \rho_{RA}^{\ox n} \Pi_t}{P(t)} \| \Pi_t \ox v_{n,|A|}\sigma_{A^n}^u)-\log|T_n^t| \\
= &-D^*_{\alpha}( \frac{\Pi_t \rho_{RA}^{\ox n} \Pi_t}{P(t)} \| \Pi_t \ox \sigma_{A^n}^u)+\log v_{n,|A|}-\log|T_n^t|,
\end{split}
\end{equation}
where the second equality is because that $\frac{\Pi_t \rho_{RA}^{\ox n} \Pi_t}{P(t)}$ is a symmetric state, we can restrict the infimum to over the symmetric states and the second inequality is from Lemma~\ref{lem:sym}.  Eq.~(\ref{equ:ev}), Eq.~(\ref{equ:1ev}) and Eq.~(\ref{equ:2ev}) give
\begin{equation}
\label{equ:new}
\begin{split}
&F^2(\tilde{\rho}_{R^nA^n}, \rho_{RA}^{\ox n}) \\
\leq & \big( \sum_{t \in \mathcal{T}_n^{\mc{X}}} P(t) 2^{\frac{1-\alpha}{2\alpha}\big\{-D^*_{\alpha}( \frac{\Pi_t \rho_{RA}^{\ox n} \Pi_t}{P(t)} \| \Pi_t \ox \sigma_{A^n}^u)+\log v_{n,|A|}-\log|T_n^t|+nr\big\}} \big)^2 \\
\leq & (n+1)^{2|R|} \max_{t \in \mathcal{T}_n^{\mc{X}}}  P(t)^2 2^{\frac{1-\alpha}{\alpha}\big\{-D^*_{\alpha}( \frac{\Pi_t \rho_{RA}^{\ox n} \Pi_t}{P(t)} \| \Pi_t \ox \sigma_{A^n}^u)+\log v_{n,|A|}-\log|T_n^t|+nr\big\}},
\end{split}
\end{equation}
where the last line is due to Eq.~(\ref{eq:typenumber}). Because Eq.~(\ref{equ:new}) holds for any $\sigma_{A^n}$ and $\tilde{\rho}_{R^nA^n}$ that  satisfy Eq.~(\ref{1}) and Eq.~(\ref{2}), according to the definition of $\epsilon^P_{\dot{R^n}:A^n}(\rho_{RA}^{\otimes n}, nr)$, we have 
\begin{equation}
\label{equ:last}
\begin{split}
& \frac{-1}{n} \log(1-\epsilon^P_{\dot{R^n}:A^n}(\rho_{RA}^{\otimes n}, nr)) \\
\geq &\frac{-1}{n} \log \sup_{\sigma_{A^n} \in \mc{S}(A^n)}\sup_{\tilde{\rho}_{R^nA^n}:\tilde{\rho}_{R^nA^n} \leq 2^{nr} \rho_R^{\ox n}\ox \sigma_{A^n}, \tilde{\rho}_{R^n}=\rho_{R}^{\ox n}} F^2(\tilde{\rho}_{R^nA^n}, \rho_{RA}^{\ox n})
\\
\geq &\frac{1}{n}\inf_{t \in \mc{T}_n^{\mc{X}}} \big\{-2 \log P(t)+\frac{1-\alpha}{\alpha}\big(D^*_{\alpha}( \frac{\Pi_t \rho_{RA}^{\ox n} \Pi_t}{P(t)} \| \Pi_t \ox \sigma_{A^n}^u)-\log v_{n,|A|}+\log|T_n^t|-nr  \big)   \big\}-\frac{2|R|\log(n+1)}{n} \\
\geq &\frac{1}{n}\inf_{t \in \mc{T}_n^{\mc{X}}} \big\{-2 \log P(t)+\frac{1-\alpha}{\alpha}\big(D^*_{\alpha}( \frac{\Pi_t \rho_{RA}^{\ox n} \Pi_t}{P(t)} \| \Pi_t \ox \sigma_{A^n}^u)+\log|T_n^t|-nr  \big)   \big\}-\frac{\log v_{n,|A|}}{n}-\frac{2|R|\log(n+1)}{n} \\
\geq &\frac{1}{n} \inf_{t \in \mc{T}_n^{\mc{X}}} \frac{1-\alpha}{\alpha} \big\{\frac{\alpha}{\alpha-1} \log P(t)+D^*_{\alpha}( \frac{\Pi_t \rho_{RA}^{\ox n} \Pi_t}{P(t)} \| \Pi_t \ox \sigma_{A^n}^u) \big\}+
\frac{1}{n} \inf_{t \in \mc{T}_n^{\mc{X}}} \frac{1-\alpha}{\alpha} \big\{\frac{\alpha}{\alpha-1} \log P(t)+\log|T_n^t|\big\}\\
&-\frac{1-\alpha}{\alpha}r-\frac{\log v_{n,|A|}}{n}-\frac{2|R|\log(n+1)}{n}.
\end{split}
\end{equation}
Eq.~(\ref{equ:last}) and Lemma~\ref{lem:li} in Appendix implies that 
\begin{equation}
\label{equ:new2}
\lim_{n \rightarrow \infty} \frac{-1}{n} \log(1-\epsilon^P(\rho_{RA}^{\otimes n}, nr)) \geq \frac{1-\alpha}{\alpha}\big\{H_{\beta}(R)_\rho-H^*_{\alpha}(R|A)_\rho-r  \big\}.
\end{equation}
Noticing that Eq.~(\ref{equ:new2}) holds for any $\alpha \in (\frac{1}{2},1)$, we have
\begin{equation}
\lim_{n \rightarrow \infty} \frac{-1}{n} \log (1-\epsilon^P(\rho_{RA}^{\otimes n}, nr)) \geq \sup_{\frac{1}{2}<\alpha<1}\frac{1-\alpha}{\alpha}\big\{H_{\beta}(R)_\rho-H^*_{\alpha}(R|A)_\rho-r  \big\}.
\end{equation}
\end{proof}

From Proposition~\ref{propo:main} and the duality relation of the sandwiched conditional R\'enyi  entropy~(Proposition~\ref{prop:mainpro}~(\romannumeral3)), we can obtain the following corollary.
\begin{corollary}
\label{cor:conver}
For pure state $\rho_{RA} \in \mc{S}(RA)$, $r \in \mathbb{R}$ and $\Delta \in \{d, P\}$, we have
\begin{equation}
\label{equ:genebound}
\lim_{n \rightarrow \infty} \frac{-1}{n} \log(1-\epsilon^\Delta_{\dot{R^n}:A^n}(\rho_{RA}^{\otimes n}, nr)) \geq \sup_{\beta>1} \frac{\beta-1}{\beta}\big\{2H_{\beta}(R)_\rho-r \big\}.
\end{equation}
\end{corollary}


\subsection{Achievability}

In this section, we derive the achievability part of Eq.~(\ref{equ:puremutual}).  In order to do this, we need to establish a relation between the partially smoothed mutual max-information and quantum data compression. Firstly, we introduce classical data compression and quantum data compression, respectively.

\emph{(a) Classical data compression} Let $\mc{X}$ be a finite alphabet set. In the procedure of classical data compression, a classical information source first produces a sequence of signals $x^n \in \mc{X}^{\times n}$. These signals are chosen from a probability 
distribution $p \in  \mc{Q}(\mc{X})$ independently. Then the sender applies an encoder function $f:\mc{X}^{\times n} \rightarrow \{0,1\}^{\times m}$ to encode $x^n$ into a binary string and sends it to the receiver. Finally, the receiver recovers the original sequence $x^n$ by a decoder function $g: \{0,1\}^{\times m} \rightarrow \mc{X}^{\times n}$.
The 2-tuple $\mc{C}_n\equiv (f,g)$ is called a $n$-length  classical data compression scheme for information source $p$. The size of $\mc{C}_n$ is $|\mc{C}_n|:=2^m$ and the success probability of $\mc{C}_n$ is given by 
\begin{equation}
\label{equ:cladef}
P_s(\mc{C}_n,p):=\sum_{x^n \in G}  p^n(x^n),
\end{equation}
where $G:=\{x^n: g\circ f(x^n)=x^n  \}.$ Shannon has  proved that to achieve asymptotically perfect data compression in which $P_s(\mc{C}_n,p) \rightarrow 1$, the compression rate $r$ must be larger than the Shannon entropy $H(p)$~\cite{Shannon1948math}. 
When the compression rate $r$ is lower than $H(p)$, the optimal success probability among all $n$-length  data compression schemes with size $2^{nr}$ goes to $0$ inevitably  , as $n \rightarrow \infty$. The exact exponential rate 
of this decay is called the strong converse exponent of classical data compression and is formally defined as
\begin{equation}
E_{sc}^{c}(p,r):=\lim_{n \rightarrow \infty} \frac{-1}{n} \log \max_{C_n:|C_n|=2^{nr}} P_s(\mc{C}_n,p).
\end{equation}
$E_{sc}^{c}(p,r)$ has been determined in~\cite{CsiszarKorner2011information} , i.e.,
\begin{equation}
\label{equ:strongcla}
E_{sc}^{c}(p,r)=\sup_{\beta >1}\frac{\beta-1}{\beta}\big\{H_{\beta}(p)-r\big\}.
\end{equation}

\emph{(b) Blind and visible quantum data compression} A quantum information source consists of an ensemble of pure states $\mathcal{Z}=\{p_i, \psi_i\}$, where $\{\psi_i\} \in \mc{S}(A)$ are pure states and $\{p_i\}$ is a probability 
distribution. The corresponding quantum state $\rho_A=\sum\limits_i p_i \psi_i$ is called the source state. According to the ways of encoding, quantum data compression can be divided into blind and visible.

In blind quantum data compression, the encoder knows the source state $\rho_A$ but has no knowledge about these $\{\psi_i\}$. The encoder and decoder maps are two quantum channels $\mc{E}_{A \rightarrow C}$ and $\mc{D}_{C \rightarrow A}$,
respectively. The size of such a blind compression protocol is the dimension of the compressed space $|C|$ and the performance for this protocol is characterized by the entanglement fidelity $F^2(\mc{D}_{C \rightarrow A} \circ \mc{E}_{A \rightarrow C}(\rho_{RA}), \rho_{RA})$, where $\rho_{RA}$ is a purification of $\rho_A$. 

In visible quantum data compression, the encoder knows these $\{\psi_i\}$ and the probability distribution $\{p_i\}$. In this case, the encoder can be an arbitrary map $\mc{V}:\{\psi_i\} \rightarrow \mc{S}(C)$~($\mc{V}$ even can be non-linear). The decoder is still a quantum channel $\mc{D}_{C \rightarrow A}$. The size and the performance for this visible compression protocol are $|C|$ and the ensemble average fidelity $\sum\limits_{i} p_i \tr \big(\mc{D}_{C \rightarrow A}\circ\mc{V}(\psi_i)\psi_i  \big)$.

The optimal performance among all blind and visible compression protocols with a fixed size $2^\lambda$  are denoted by $\epsilon^b(\rho_A, \lambda)$ and  $\epsilon^v(\rho_A, \lambda)$, i.e.,
\begin{align}
\begin{split}
\epsilon^b(\rho_A, \lambda):&=\max_{(\mc{D}_{C \rightarrow A},\mc{E}_{A \rightarrow C}): |C|=2^\lambda} F^2(\mc{D}_{C \rightarrow A} \circ \mc{E}_{A \rightarrow C}(\rho_{RA}), \rho_{RA}), \\
\epsilon^v(\rho_A, \lambda):&=\max_{(\mc{D}_{C \rightarrow A},\mc{V}): |C|=2^\lambda} \sum\limits_{i} p_i \tr \big(\mc{D}_{C \rightarrow A}\circ\mc{V}(\psi_i)\psi_i  \big).
\end{split}
\end{align}

The paper~\cite{schumacher1995quantum, BFJS1996general} has proved that a compression rate larger than the Von Neumann entropy $H(A)_\rho$ is necessary and sufficient for achieving asymptotically perfect blind and visible quantum data compression.
When $r<H(A)_\rho$, $\epsilon^b(\rho^{\ox n}_A, nr)$ and $\epsilon^v(\rho^{\ox n}_A, nr)$ converge to $0$ exponentially fast, and the strong converse exponents for blind and visible quantum data compression  are the exact rates of these exponential convergence, i.e.,
\begin{align}
\begin{split}
E_{sc}^{b}(\rho_A,r):=\lim_{n \rightarrow \infty} \frac{-1}{n} \log \epsilon^b(\rho^{\ox n}_A, nr), \\
E_{sc}^{v}(\rho_A,r):=\lim_{n \rightarrow \infty} \frac{-1}{n} \log \epsilon^v(\rho^{\ox n}_A, nr).
\end{split}
\end{align}
From Eq.~(\ref{equ:strongcla}), we can obtain the following result for $E_{sc}^{b}(\rho_A,r)$.

\begin{proposition}
\label{thm:qdata}
For any quantum information source $\mc{Z}=\{p_i,\psi_i\}$  with source state $\rho_{A}=\sum\limits_{x \in \mc{X}} p(x)\proj{x}$ and $r\geq 0$, we have
\begin{equation}
\label{equ:re}
E_{sc}^{b}(\rho_A,r) \leq \sup_{\beta >1}\frac{\beta-1}{\beta}\big\{2H_{\beta}(A)_\rho-2r\big\}.
\end{equation}
\end{proposition}

\begin{proof}
For the classical information source $p:=\{p(x)\}_{x\in \mc{X}}$ and the classical compression rate $r$, we set the sequence of $n$-length classical data compression schemes which achieve the strong converse exponent in Eq.~(\ref{equ:strongcla}) as $\{(f_n, g_n)\}_{n \in \mathbb{N}}$ and $S_n':=\{x^n~|~g_n \circ f_n(x^n)=x^n \}$. It is obvious that $|S_n'| \leq 2^{nr}$, so we can 
add some signals into $S_n'$ to get a new set $S_n$ with $|S_n|=2^{nr}$. Now, we let $\Pi_n:=\sum\limits_{x^n \in S_n} \proj{x^n}$ and construct an encoder channel $\mc{E}_n$ as $\mc{E}_n(X)= \Pi_n X \Pi_n+ (\tr(I-\Pi_n)X) \frac{\Pi_n}{2^{nr}}$, a decoder channel $\mc{D}_n$ as the identity channel. Then, for such a quantum data compression protocol $(\mc{E}_n,\mc{D}_n)$, we can bound the performance of it as 
\begin{equation}
\label{equ:con}
\begin{split}
&F^2(\mc{D}_n \circ \mc{E}_n(\rho_{RA}^{\ox n}), \rho_{RA}^{\ox n})\\
\geq & F^2(\Pi_n \rho_{RA}^{\ox n} \Pi_n, \rho_{RA}^{\ox n})\\
=& (\tr \Pi_n \rho_{A}^{\ox n})^2\\
\geq &(\sum_{x^n \in S_n' } p^n(x^n))^2.
\end{split}
\end{equation}
Eq.~(\ref{equ:strongcla}) and Eq.~(\ref{equ:con}) imply that 
\begin{equation} 
E_{sc}^{b}(\rho_A,r) \leq \sup_{\beta >1}\frac{\beta-1}{\beta}\big\{2H_{\beta}(A)_\rho-2r\big\}.
\end{equation}
\end{proof}
With the above tools in hand, we can prove the achievability part of Eq.~(\ref{equ:puremutual}). 

\begin{proofof}[of the achievability part of Eq.~(\ref{equ:puremutual})]
 By Lemma~\ref{lem:opein} in Appendix, we have that for any quantum channel $\mc{E}_{A \rightarrow C}$ and 
$\mc{D}_{C \rightarrow A}$, it holds that

\begin{equation}
\label{equ:dap}
\mc{D}_{C \rightarrow A} \circ \mc{E}_{A \rightarrow C}(\rho_{RA})\leq  |C|^2 \rho_R \ox \mc{D}_{C \rightarrow A}(\frac{I_C}{|C|}).
\end{equation}
Eq.~(\ref{equ:dap}) implies that 
\begin{equation}
\label{equ:relation}
\epsilon^b(\rho_A, r)\leq \sup_{\sigma_A \in \mc{S}(A)} \sup_{\tilde{\rho}_{RA}:\tilde{\rho}_{RA}\leq 2^{2r} \rho_R \ox \sigma_A}F^2(\tilde{\rho}_{RA},\rho_{RA}).
\end{equation}
From Eq.~(\ref{equ:re}) and Eq.~(\ref{equ:relation}), we can get
\begin{equation}
\begin{split}
&\lim_{n \rightarrow \infty} \frac{-1}{n} \log \big(1-\epsilon^P_{\dot{R^n}:A^n}(\rho_{RA}^{\ox n},nr)\big)  \\
\leq &\lim_{n \rightarrow \infty} \frac{-1}{n} \log \sup_{\sigma_{A^n} \in \mc{S}(A^n)} \sup_{\tilde{\rho}_{R^nA^n}:\tilde{\rho}_{R^nA^n}\leq 2^{nr} \rho^{\ox n}_R \ox \sigma_{A^n}}\frac{1}{2}F^2(\tilde{\rho}_{R^nA^n},\rho^{\ox n}_{RA}) \\
\leq&\lim_{n \rightarrow \infty} \frac{-1}{n} \log \frac{1}{2}\epsilon^b(\rho_{A}^{\ox n},\frac{1}{2}nr)  \\
=& \sup_{\beta >1}\frac{\beta-1}{\beta}\big\{2H_{\beta}(R)_\rho-r\big\},
\end{split}
\end{equation}
where in the last line, we use the fact $H_\beta(R)_\rho=H_\beta(A)_\rho$ when $\rho_{RA}$ is a pure state.
\end{proofof}


\section{Applications}
\label{sec:application}

\subsection{Quantum data compression}

In \cite{Hayashi2002exponents}, Hayashi studied the strong converse exponent of quantum data compression. He established the exact strong converse exponent of  visible quantum data compression. However, the exact strong converse exponent of blind quantum data compression is still unknown. In this section, we use the results in the past sections to determine the strong converse exponent of blind quantum data compression, thus filling the gap in previous work.  Our main result is as follows.

\begin{theorem}
\label{thm:datacom}
For any quantum information source $\mc{Z}=\{p_i, \psi_i\}$ with source state  $\rho_{A}$ and $r\geq 0$, we have
\begin{align}
E_{sc}^{b}(\rho_A,r)& = \sup_{\beta >1}\frac{\beta-1}{\beta}\big\{2H_{\beta}(A)_\rho-2r\big\} \label{equ:final}.
\end{align}
\end{theorem}
\begin{proof}
The achievability part of Eq.~(\ref{equ:final}) has been proved in Proposition~\ref{thm:qdata}. The optimality part can be derived directly from Eq.~(\ref{equ:relation}) and Corollary~\ref{cor:conver}.
\end{proof}

\begin{remark}
In the proof of Theorem~\ref{thm:datacom}, we mainly prove the optimality part. This part also can be 
deduced from \cite[Lemma~3.5]{LWD2016strong} in which the authors consider the more general scenario:
quantum state merging.
\end{remark}


\subsection{Intrinsic randomness}

Intrinsic randomness problem is a basic information task in which we aim to use a function $f: \mc{X} \rightarrow \mc{Z}$ to extract as much uniform random number as possible from a given information source distributed as
$p \in \mc{Q}(\mc{X})$. This problem can be regarded as a special case of privacy amplification in which there is no adversary. Hence, as introduced in subsection~\ref{subsec:privacy}, the probability distribution $p$ can be 
regarded as a quantum state $p=\sum\limits_{x \in \mc{X}}p(x)\proj{x}$ and the action of the function $f$ can be expressed as a quantum channel
\[
\Phi_f : p \mapsto \sum_{x\in \mc{X}} \ket{f(x)}\bra{x} p \ket{x}\bra{f(x)}.
\]
The performance of the protocol is characterized by the trace distance between the final state $\Phi_f(p)$ and the ideal state $\frac{I_Z}{|\mc{Z}|}$ and the size of the protocol is $|\mc{Z}|$.

In asymptotic setting where an arbitrary large number of copies of $p$ is available, Vembu and Verd\'u~\cite{VembVerdu2002generating} proved that the suprema of achievable uniform random number generation rates equals to $H(p)$.  In~\cite{Hayashi2008second} and~\cite{Hayashi2013tight}, Hayashi  established the second-order asymptotic and the error exponent of this problem, respectively.

When the uniform random number generation rate is larger than $H(p)$, the optimal performance among all protocols with size $2^{nr}$ converge to $1$ as $n$ tends to infinity. The corresponding exponential decay rate is called 
the strong converse exponent of intrinsic randomness and is formally defined as
\[
E^{\rm{ir}}_{sc}(p,r):=\lim_{n \rightarrow \infty} \frac{-1}{n} \log (1-\inf_{f_n\in \mc{F}_n(r)}d(\Phi_{f_n}(p^{\ox n}),\frac{I_{Z^n}}{|\mc{Z}^n|})),
\]
where $\mc{F}_n(r)$ is the set of functions from $\mc{X}^{\times n}$ to $\mc{Z}_n=\{1, \ldots, 2^{nr}\}$.

In this subsection, we use the strong converse exponent for partially smoothing of the conditional min-entropy to derive the exact expression of $E^{\rm{ir}}_{sc}(p,r)$. Our main result is the following theorem.

\begin{theorem}
\label{thm:pritra}
Let $p \in \mathcal{Q}(\mc{X})$ and $r \geq 0$, it holds that
\begin{equation}
\label{equ:priclatra}
E^{\rm{ir}}_{sc}(p,r)=\sup_{0 \leq \alpha \leq 1} (1-\alpha)\big\{r-H_\alpha(p) \big\}.
\end{equation}
\end{theorem}

We start the proof of Theorem~\ref{thm:pritra} with the establishment of the optimality part.

\begin{proofof}[of the optimality part]
Lemma~\ref{lem:monohash} in Appendix implies that for any function $f_n~:~\mc{X}^{\times n}\rightarrow \mc{Z}_n=\{1,\ldots,2^{nr}\}$, we have
\begin{equation}
\label{equ:hatotra}
 H^{\epsilon,d}_{\rm{min}}(X^n)_{p^{\ox n}} \geq  H^{\epsilon,d}_{\rm{min}}(Z_n)_{\Phi_{f_n}(p^{\ox n})}.
\end{equation}
Eq.~(\ref{equ:hatotra}) gives 
\begin{equation}
\label{equ:hotoretra}
 \epsilon^d_{Z_n|\dot{\emptyset}}(\Phi_{f_n}(p^{\ox n}), nr) \geq \epsilon^d_{X^n|\dot{\emptyset}}(p^{\ox n},  nr),
\end{equation}
where $\emptyset$ denotes the trivial system. It is obvious that
\begin{equation}
\label{equ:hato1tra}
d(\Phi_{f_n}(p^{\ox n}),\frac{I_{Z_n}}{|\mc{Z}_n|}) \geq \epsilon^d_{Z^n|\dot{\emptyset}}(\Phi_{f_n}(p^{\ox n}), nr).
\end{equation}
From Eq.~(\ref{equ:hotoretra}) and Eq.~(\ref{equ:hato1tra}), we can obtain
\begin{equation}
\label{equ:hatofianltra}
\inf_{f_n \in \mc{F}_n(r)}d(\Phi_{f_n}(p^{\ox n}, \frac{I_{Z_n}}{|\mc{Z}_n|})
\geq \epsilon^d_{X^n|\dot{\emptyset}}(p^{\ox n},  nr).
\end{equation}
The optimality part follows from Theorem~\ref{thm:main1} together with Eq.~(\ref{equ:hatofianltra}).

\end{proofof}
Before we move to the proof of the achievability part, we first establish an equivalent variational expression of the strong converse exponent. 

\begin{proposition}
\label{prop:pritracla1}
For $p \in \mc{Q}(\mc{X})$ and $r \geq 0$, we have
\begin{equation}
\sup_{0 \leq \alpha \leq 1} (1-\alpha)\big\{r-H_\alpha(p) \big\}=\inf_{q \in \mc{Q}(\mc{X})} \big\{D(q\|p)+|r-H(q)-D(q\|p)|^+ \big\}.
\end{equation}
\end{proposition}

\begin{proof}
By Proposition~\ref{prop:mainpro}~(\romannumeral11) , we have
\begin{equation}
\begin{split}
&\sup_{0 \leq \alpha \leq 1} (1-\alpha)\big\{r-H_\alpha(p) \big\} \\
=&\sup_{0 < \alpha <1} (1-\alpha)\big\{r-H_\alpha(p) \big\} \\
=&\sup_{0< \alpha<1}\inf_{q \in  \mc{Q}(\mc{X})}(1-\alpha) \big\{ r+D(q\|I_X)+\frac{\alpha}{1-\alpha}D(q\|p) \big\} \\
=&\sup_{0< \alpha<1}\inf_{q \in  \mc{Q}(\mc{X})} \big\{D(q\|p)+(1-\alpha)\{r-H(q)-D(q\|p)  \}   \big\} \\
=&\inf_{q \in  \mc{Q}(\mc{X})}\sup_{0< \alpha<1} \big\{D(q\|p)+(1-\alpha)\{r-H(q)-D(q\|p)  \}   \big\} \\
=&\inf_{q \in  \mc{Q}(\mc{X})} \big\{D(q\|p)+|r-H(q)-D(q\|p)|^+ \big\},
\end{split}
\end{equation}
where the fifth line comes from Sion's minimax theorem.
\end{proof}

Now, we prove the achievability part of Theorem~\ref{thm:pritra}.

\begin{proofof}[of the achievability part]
If $r\leq H(p)$, it is obvious that all terms in Eq.~(\ref{equ:priclatra}) equal to $0$ and the assertion 
holds trivially. Hence, for the rest we assume that $r > H(p)$.

Because $r>H(p)$, there exists $\delta>0$ and an integer $N_\delta$ such that
\begin{equation}
\label{equ:nearcon}
D(q\|p) \geq \delta,~\forall q \in \mc{Q}(\mc{X})~\text{and}~H(q)\geq r,
\end{equation}
and
\begin{equation}
\label{equ:largecon}
2^{nr-n\delta}\leq 2^{nr}-1,~\forall n \geq N_\delta.
\end{equation}
We claim that for $n \geq N_\delta$, it holds that
\begin{equation}
\label{equ:rquitypr}
\frac{-1}{n} \log (1-\inf_{f_n \in \mc{F}_n(r)}d(\Phi_{f_n}(p^{\ox n}),\frac{I_{Z_n}}{|\mc{Z}_n|})) \leq \inf_{t \in \mc{T}_n^{\mc{X}}} \big\{D(t\|p)+|r-H(t)-D(t\|p)|^+ \big\}+\frac{|\mc{X}|}{n}\log(n+1).
\end{equation}
We divide the proof of Eq.~(\ref{equ:rquitypr}) into three cases according to the range of $t$.

\textbf{Case 1:}  $|T_n^t| < 2^{nr}\land r<H(t)+D(t\|p)$; 

\noindent We assign an order for the sequence in $\mc{X}^{\times n}$ and let the first $|T_n^t|$ items be the sequences in $T_n^t$, i.e., $\mc{X}^{\times n}=\{x^n_1,\ldots,x^n_{|T_n^t|},\ldots,x^n_{2^{nr}} \}$.
Then we construct a function $f_n^1: \mc{X}^{\times n} \rightarrow \mc{Z}_n=\{1,\ldots,2^{nr}\}$ as
\begin{equation}
f^1_n(x^n_i)= \begin{cases}
i & 1\leq i \leq |T_n^t|, \\
i+1                        & i>|T_n^t|.
                  \end{cases}
\end{equation}
The performance of $f^1_n$ can be evaluated as
\begin{equation}
\label{equ:case1}
\begin{split}
&1-d(\Phi_{f^1_n}(p^{\ox n}),\frac{I_{Z_n}}{|\mc{Z}_n|}) \\
=&1-\tr(\Phi_{f^1_n}(p^{\ox n})-\frac{I_{Z_n}}{|\mc{Z}_n|})_+ \\
=& \tr \Phi_{f^1_n}(p^{\ox n})\{\Phi_{f^1_n}(p^{\ox n})\leq \frac{I_{Z_n}}{|\mc{Z}_n|}  \}+\tr \frac{I_{Z_n}}{|\mc{Z}_n|}\{ \Phi_{f^1_n}(p^{\ox n})>\frac{I_{Z_n}}{|\mc{Z}_n|} \} \\
\geq&\tr \Phi_{f^1_n}(p^{\ox n})\{\Phi_{f^1_n}(p^{\ox n})\leq \frac{I_{Z_n}}{|\mc{Z}_n|}  \}\\
=&\sum_{i=1}^{|T_n^t|}p^n(x^n_i) \\
\geq&(n+1)^{-|\mc{X}|}2^{-nD(t\|p)},
\end{split}
\end{equation}
where the fifth line is because that $r<H(t)+D(t\|p)$ implies $p^n(x^n)<2^{-nr},~\forall x^n \in T_n^t$ and the last line is from Eq.~(\ref{eq:proset}).
Eq.~(\ref{equ:case1}) gives
\begin{equation}
\begin{split}
&\frac{-1}{n} \log(1-\inf_{f_n \in \mc{F}_n(r)}d(\Phi_{f_n}(p^{\ox n}),\frac{I_{Z_n}}{|\mc{Z}_n|})) \\
\leq&\frac{-1}{n} \log(1-d(\Phi_{f^1_n}(p^{\ox n}),\frac{I_{Z_n}}{|\mc{Z}_n|})) \\
\leq&D(t\|p)+\frac{|\mc{X}|}{n}\log(n+1) \\
=&D(t\|p)+|r-H(t)-D(t\|p)|^+ +\frac{|\mc{X}|}{n}\log(n+1).
\end{split}
\end{equation}

\textbf{Case 2:}  $|T_n^t| \geq 2^{nr}\land r<H(t)+D(t\|p)$;

\noindent We let $m_n=\lfloor 2^{nH(t)+nD(t\|p)-nr}\rfloor$ and $k_n=\lfloor \frac{|T_n^t|}{m_n} \rfloor$.  Condition ``$|T_n^t| \geq 2^{nr}$'' and Eq.~(\ref{eq:numt}) give $r \leq H(t)$.  Hence, from Eq.~(\ref{eq:numt}), Eq.~(\ref{equ:nearcon}) and Eq.~(\ref{equ:largecon}), we have
\begin{equation}
k_n \leq \frac{2^{nH(t)}}{2^{nH(t)+nD(t\|p)-nr}} \leq 2^{nr-n\delta} \leq 2^{nr}-1.
\end{equation}
We assign the same order for the sequence in $\mc{X}^{\times n}$ as in \textbf{Case 1}. Then, we can construct a function $f_n^2: \mc{X}^{\times n} \rightarrow \mc{Z}_n=\{1,\ldots,2^{nr}\}$ as
\begin{equation}
f^2_n(x^n_i)= \begin{cases}
\lceil \frac{i}{m_n}\rceil & 1\leq i \leq k_nm_n, \\
k_n+1                        & i>k_nm_n.
                  \end{cases}
\end{equation}
By the construction of $f_n^2$, it is esay to see that 
\begin{equation}
\Phi_{f_n^2}(p^{\ox n})(j)=\sum^{m_n j}_{i=m_n(j-1)+1} p^n(x^n_i)=m_n 2^{-nH(t)-nD(t\|p)} \leq 2^{-nr},~\forall 1 \leq j \leq k_n,
\end{equation}
and
\begin{equation}
\Phi_{f_n^2}(p^{\ox n})(k_n+1)\geq 2^{-nr}\geq m_n 2^{-nH(t)-nD(t\|p)}  \geq \sum_{i=m_nk_n+1}^{|T_n^t|}p^n(x^n_i).
\end{equation}
\end{proofof}
Hence, we can evaluate the performance of $f_n^2$ as
\begin{equation}
\label{equ:case2}
\begin{split}
&1-d(\Phi_{f^2_n}(p^{\ox n}),\frac{I_{Z_n}}{|\mc{Z}_n|}) \\
=& \tr \Phi_{f^2_n}(p^{\ox n})\{\Phi_{f^2_n}(p^{\ox n})\leq \frac{I_{Z_n}}{|\mc{Z}_n|}  \}+\tr \frac{I_{Z_n}}{|\mc{Z}_n|}\{ \Phi_{f^2_n}(p^{\ox n})>\frac{I_{Z_n}}{|\mc{Z}_n|} \} \\
=&\sum_{i=1}^{m_nk_n}p^n(x^n_i)+2^{-nr} \\
\geq& \sum_{i=1}^{m_nk_n}p^n(x^n_i)+\sum_{j=m_nk_n+1}^{|T_n^t|}p^n(x^n_j)\\
=&\sum_{x^n\in T_n^t} p^n(x^n) \\
\geq&(n+1)^{-|\mc{X}|}2^{-nD(t\|p)}.
\end{split}
\end{equation}
Eq.~(\ref{equ:case2}) implies that Eq.~(\ref{equ:rquitypr}) holds in this case.

\textbf{Case 3:}  $r \geq H(t)+D(t\|p)$;

If $t=p$, then $r>H(p)=H(t)$. If $t \neq p$, then $r \geq H(t)+D(t\|p)>H(t)$. Hence, in this case, we have
\begin{equation}
2^{nr}>2^{nH(t)} \geq |T_n^t|,
\end{equation}
where the last inequality comes from Eq.~(\ref{eq:numt}). Next, we assign the same order for the sequence in $\mc{X}^{\times n}$ as in \textbf{Case 1} and define a function $f_n^3: \mc{X}^{\times n} \rightarrow \mc{Z}_n=\{1,\ldots,2^{nr}\}$ as
\begin{equation}
f^3_n(x^n_i)= \begin{cases}
i & 1\leq i \leq |T_n^t|, \\
|T_n^t|+1                       & i > |T_n^t|.
                  \end{cases}
\end{equation}
When $1\leq i \leq |T_n^t|$, we have
\begin{equation}
\label{equ:f3}
\Phi_{f^3_n}(p^{\ox n})(i)=p^n(x^n_i)=2^{-nH(t)-nD(t\|p)}\geq 2^{-nr}.
\end{equation}
By Eq.~(\ref{equ:f3}), we can obtain
\begin{equation}
\label{equ:case3}
\begin{split}
&1-d(\Phi_{f^3_n}(p^{\ox n}),\frac{I_{Z_n}}{|\mc{Z}_n|}) \\
=&1-\tr(\Phi_{f^3_n}(p^{\ox n})-\frac{I_{Z_n}}{|\mc{Z}_n|})_+ \\
=& \tr \Phi_{f^3_n}(p^{\ox n})\{\Phi_{f^3_n}(p^{\ox n})< \frac{I_{Z_n}}{|\mc{Z}_n|}  \}+\tr \frac{I_{Z_n}}{|\mc{Z}_n|}\{ \Phi_{f^3_n}(p^{\ox n})\geq \frac{I_{Z_n}}{|\mc{Z}_n|} \} \\
\geq& \tr \frac{I_{Z_n}}{|\mc{Z}_n|}\{ \Phi_{f^3_n}(p^{\ox n})\geq \frac{I_{Z_n}}{|\mc{Z}_n|} \} \\
\geq & \frac{|T_n^t|}{|\mc{Z}_n|} \\
\geq&(n+1)^{-|\mc{X}|}2^{nH(t)-nr},
\end{split}
\end{equation}
where the last line is due to Eq.~(\ref{eq:numt}). Eq.~(\ref{equ:case3}) leads to 
\begin{equation}
\begin{split}
&\frac{-1}{n} \log(1-\inf_{f_n \in \mc{F}_n(r)}d(\Phi_{f_n}(p^{\ox n}),\frac{I_{Z_n}}{|\mc{Z}_n|})) \\
\leq&\frac{-1}{n} \log(1-d(\Phi_{f^3_n}(p^{\ox n}),\frac{I_{Z_n}}{|\mc{Z}_n|})) \\
\leq&r-H(t)+\frac{|\mc{X}|}{n}\log(n+1) \\
=&D(t\|p)+|r-H(t)-D(t\|p)|^+ +\frac{|\mc{X}|}{n}\log(n+1).
\end{split}
\end{equation}
This completes the proof of \textbf{Case 3}. 

By letting $n \rightarrow \infty$ in Eq.~(\ref{equ:rquitypr}), we have
\begin{equation}
\begin{split}
&E^{\rm{ir}}_{sc}(p,r)\\
=&\lim_{n \rightarrow \infty}\frac{-1}{n} \log (1-\inf_{f_n \in \mc{F}_n(r)}d(\Phi_{f_n}(p^{\ox n}),\frac{I_{Z_n}}{|\mc{Z}_n|}))  \\
\leq &\inf_{q \in \mc{Q}({\mc{X}})} \big\{D(t\|p)+|r-H(t)-D(t\|p)|^+ \big\}\\
=&\sup_{0 \leq \alpha \leq 1} (1-\alpha)\big\{r-H_\alpha(p) \big\},
\end{split}
\end{equation}
where the last line is due to Proposition~\ref{prop:pritracla1}.


\subsection{Classical state splitting}

For a random variable possessed by Alice and distributed as $P_{XY}\in \mc{Q}(\mc{X}\mc{Y})$, a classical state splitting protocol consists of using free shared randomness between Alice and Bob which is independent of $XY$, sending classical message from Alice to Bob, and at last outputting random
variable $\hat{Y}$ at Bob's side. The performance of a protocol is characterized by the trace distance between $P_{XY}$ and $P_{X\hat{Y}}$. We are interested in the trade-off between the performance and the classical communication cost in a protocol. The minimal number of bits communicated from Alice to Bob to achieve classical state splitting with error $\epsilon \in (0,1)$ is denoted as $R(P_{XY},\epsilon)$. We denote the inverse function of $R(P_{XY},\epsilon)$ as $\epsilon(P_{XY},\lambda)$ which represents the optimal performance among all protocols with classical communication cost $2^{\lambda}$, i.e.,
\[
\epsilon(P_{XY},\lambda):=\min\{\epsilon~|~R(P_{XY},\epsilon)\leq \lambda   \}.
\]

In~\cite{AJW2017unified}, the authors proved that the minimal classical communication cost rate to achieve asymptotically perfect 
state splitting equals to the mutual information $I(X:Y)_P$. The reference~\cite{ABJT2020partially}  further strengthens this result to second-order asymptotics.

Classical state splitting satisfies the strong converse property, i.e., when $r<I(X:Y)_P$, $\epsilon(P_{XY}^{\ox n}, nr)$ converges to $1$ exponentially fast and the 
exact exponent of this decay is called the strong converse exponent and is defined as
\begin{equation}
    E^{\rm{spl}}_{sc}(P_{XY},r):=\lim_{n \rightarrow \infty}\frac{-1}{n} \log (1-\epsilon(P_{XY}^{\ox n}, nr)).
\end{equation}
 In this subsection, we determine the exact expression of $E^{\rm{spl}}_{sc}(P_{XY},r)$. The main result is
\begin{theorem}
For $P_{XY} \in \mc{Q}(\mc{X}\mc{Y})$ and $r \geq 0$, it holds that
\begin{equation}
 E^{\rm{spl}}_{sc}(P_{XY},r)= \sup_{0 \leq \alpha \leq 1}
(1-\alpha)\big\{I_\alpha(X:Y)_\rho-r\big\}.
\end{equation}
\end{theorem}
\begin{proof}
We first deal with the optimality part. By~\cite[Theorem 4]{ABJT2020partially}, for any $n \in \mathbb{N}$ and $\epsilon \in (0,1)$, we have
    \begin{equation}
    \label{equ:split}
        I^{\epsilon,d}_{\rm{max}}(\dot{X}^n:Y^n)_{P^{\ox n}}\leq R(P^{\ox n}_{XY},\epsilon).
    \end{equation}
 Eq.~(\ref{equ:split}) implies that 
\begin{equation}
\label{equ:spliimp}
\epsilon(P_{XY}^{\ox n}, nr) \geq \epsilon^d_{\dot{X}^n,Y^n}(P_{XY}^{\ox n}, nr).
\end{equation}
\end{proof}
The optimality part follows from Eq.~(\ref{equ:spliimp}) and Theorem~\ref{thm:main1}.  

In the derivation of the achievability part, we use the same protocol as in the proof of~\cite[Theorem 4]{ABJT2020partially}. For any $\delta>0$, there exists a sufficiently large integer $N_\delta$ such that when $n \geq N_\delta$,
\begin{equation}
\label{equ:rate}
    2^{n(r-\delta)} \leq \frac{2^{nr}-1}{n}.
\end{equation}
For $n \geq N_\delta$, we let $K_n:=\log(2^{nr}-1)-\log\log\frac{1}{2^{-n}}$, $R_n:=K_n+\log\log\frac{1}{2^{-n}}$. Then, for any $P'_{X^nY^n} \in \mc{Q}(\mc{X}^{\times n}\mc{Y}^{\times n})$ and $q_{Y^n} \in \mc{Q}(\mc{Y}^{\times n})$ which satisfy
\begin{equation}
\label{equ:splitcon}
P'_{X^nY^n} \leq 2^{K_n} P^{\ox n}_{X} \ox q_{Y^n},~P'_{X^n}=P_{X}^{\ox n},
\end{equation}
 the protocol can be stated below.

\textbf{Protocol}: Alice and Bob share the random variable $P_{X}^{\ox n} \ox q^{1}_{Y^n}\ox \ldots\ox q^{2^{R_n}}_{Y^n}$, where Alice holds $P_X^{\ox n}$, 
$q^{1}_{Y^n}, \ldots, q^{2^{R_n}}_{Y^n}$ are the copies of $q_{Y^n}$ and act as the shared randomness between Alice and Bob. When Alice obtain a sample $x$ from $P^n_X(x^n)$, they proceed in the following steps.
\begin{enumerate}
\item Alice sets $i=1$.

\item While $i \leq 2^{R_n}$:

\item Alice takes a sample from $q_{Y^n}^i$.

\item With probability $\frac{P'_{X^nY^n}(y^n|x^n)}{2^{K_n}q_{Y^n}(y^n)}$ she accepts the sample, sends $i$ to Bob and exits the whole loop.

\item With probability  $1-\frac{P'_{X^nY^n}(y^n|x^n)}{2^{K_n}q_{Y^n}(y^n)}$ she updates $i \rightarrow i+1$
and goes to Step 2~(End While),

\item If $i>2^{R_n}$, Alice sends $2^{R_n}+1$ to Bob.

\item Bob receives Alice's message, which we call $j$.
If $j>2^{R_n}$, Bob outputs a sample distributed as 
$q_{Y^n}$. Else he outputs the sample from $q_{Y^n}^j$.
\end{enumerate}
Based on the same arguments as in the proof of~\cite[Theorem 4]{ABJT2020partially}, the number of bits communicated is $\log(2^{R_n}+1)=nr$ and the output distribution is equal to
\[
P''_{X^nY^n}=\sum_{x^n \in \mc{X}^{\times n}
}P^n_{X}(x^n)\proj{x^n} \ox ((1-\lambda_n)P'_{X^nY^n}(\cdot|x^n)+\lambda_n q_{Y^n}),
\]
where $\lambda_n$ is the probability that Alice does not find any sample and it can be bounded as
\begin{equation}
\label{equ:split1}
  \lambda_n=(1-2^{-K_n})^{2^{R_n}}\leq (2^{-2^{-K_n}})^{2^{R_n}} =2^{-2^{\log\log \frac{1}{2^{-n}}}}=2^{-n}.
\end{equation}
Hence, the performance of this protocol can be evaluated as
\begin{equation}
\label{equ:split2}
\begin{split}
&d(P''_{X^nY^n},P_{XY}^{\ox n}) \\
=&\sum_{x^n \in \mc{X}^{\times n}}P_X^n(x^n) d(P''_{X^nY^n}(\cdot|x^n),P^{\ox n}_{XY}(\cdot|x^n))\\
=&\sum_{x^n \in \mc{X}^{\times n}}P_X^n(x^n) d((1-\lambda_n)P'_{X^nY^n}(\cdot|x^n)+\lambda_n q_{Y^n},P^{\ox n}_{XY}(\cdot|x^n))\\
\leq&(1-\lambda_n) \sum_{x^n \in \mc{X}^{\times n}}P_X^n(x^n) d(P'_{X^nY^n}(\cdot|x^n),P^{\ox n}_{XY}(\cdot|x^n))+\lambda_n\sum_{x^n \in \mc{X}^{\times n}}P_X^n(x^n) d(q_{Y^n},P^{\ox n}_{XY}(\cdot|x^n)) \\
\leq &(1-\lambda_n)d(P'_{X^nY^n},P_{XY}^{\ox n})+\lambda_n
\end{split}
\end{equation}
Eq.~(\ref{equ:split1}) and Eq.~(\ref{equ:split2}) give
\begin{equation}
\label{equ:split3}
\begin{split}
&1-\epsilon(P^{\ox n}_{XY},nr)\\
\geq &1-d(P''_{X^nY^n},P_{XY}^{\ox n}) \\
\geq &(1-\lambda_n)(1-d(P'_{X^nY^n},P_{XY}^{\ox n})) \\
\geq &(1-2^{-n}) (1-d(P'_{X^nY^n},P_{XY}^{\ox n})).
\end{split}
\end{equation}
Because Eq.~(\ref{equ:split3}) holds for any $P'_{X^nY^n}$
and $q_{Y^n}$ which satisfy Eq.~(\ref{equ:splitcon}), we have
\begin{equation}
\label{equ:splitfinal}
  1-  \epsilon(P^{\ox n}_{XY},nr) \geq (1-2^{-n})(1-\epsilon^d_{\dot{X}^n:Y^n}(P_{XY}^{\ox n}, K_n))\geq (1-2^{-n})(1-\epsilon^d_{\dot{X}^n:Y^n}(P_{XY}^{\ox n}, n(r-\delta))),
\end{equation}
where the second inequality is due to Eq.~(\ref{equ:rate}). Eq.~(\ref{equ:splitfinal}) and Theorem~\ref{thm:main1} imply that
\begin{equation}
 E^{\rm{spl}}_{sc}(P_{XY},r) \leq  \sup_{0 \leq \alpha \leq 1}
(1-\alpha)\big\{I_\alpha(X:Y)_\rho-r+\delta\big\}.
\end{equation}
By letting $\delta \rightarrow 0$, we complete the proof of the achievability part.


\section{Conclusion}
\label{sec:conclu}

In this work, we establish the exact strong converse exponents of partially smoothed information measures for classical states and pure states.  The formulas of the strong converse exponents based on trace distance for classical states~(Eq.~(\ref{equ:mainl}) and Eq.~(\ref{equ:clamutual})) do not coincide with those for pure states~(Eq.~(\ref{equ:main}) and Eq.~(\ref{equ:puremutual})). Hence, the strong converse exponents of partially smoothed information measures based on trace distance are not uniform across states. As an application, we use these results to derive the strong converse exponents of quantum blind data compression, intrinsic randomness and classical state splitting.
As a by-product in establishing the main results, we determine the strong converse exponent of classical privacy amplification which we believe to be of independent interest.  Because partially smoothed information measures can be used to give tight characterizations in many one-shot  quantum information processing tasks~\cite{ABJT2020partially,FWTB2024channel}, we anticipate that our results can have more applications in quantum information field.


\section*{Acknowledgments}

We acknowledge funding by the European Research Council (ERC Grant Agreement No. 948139) and MB acknowledges support from the Excellence Cluster - Matter and Light for Quantum Computing (ML4Q).


\section{Appendix}

\subsection{Auxiliary Lemmas}

\begin{lemma}
\label{lem:li}
For $\rho_{RA}$, $\{P(t)\}_{t \in \mc{T}_n^\mc{X}}$, $\{\Pi_t\}_{t \in \mc{T}_n^\mc{X}}$, $\mc{E}_n$ defined in the proof of Proposition~\ref{propo:main} and $\frac{1}{2}<\alpha<1$, we have
\begin{align}
&\lim_{n \rightarrow \infty} \inf_{t \in \mc{T}_n^\mc{X}} \frac{1}{n} \big\{\frac{\alpha}{\alpha-1}\log P(t)+D^*_{\alpha}(\frac{\Pi_t\rho_{RA}^{\ox n}\Pi_t}{P(t)}\|\Pi_t\ox \sigma_{A^n}^u) \big\}=-H^*_{\alpha}(R|A)_\rho   
\\ 
\label{equ:confi}
&\lim_{n \rightarrow \infty}  \inf_{t \in \mc{T}_n^\mc{X}}\frac{1}{n} \big\{\frac{\alpha}{\alpha-1}\log P(t)+\log|T_n^t|\big\}=H_\beta(R)_\rho  , 
\end{align}
where $\frac{1}{\alpha}+\frac{1}{\beta}=2$.
\end{lemma}

\begin{proof}
By the data processing inequality of the sandwiched R\'enyi divergence and Proposition~\ref{prop:mainpro}~(\romannumeral4), we have
\begin{equation}
\label{equ:typecon}
\begin{split}
&D^*_{\alpha}(\mc{E}_n(\rho_{RA}^{\ox n}) \| I_{R}^{\ox n} \ox \sigma_{A^n}^u) \\
\leq &D^*_{\alpha}(\rho_{RA}^{\ox n} \| I_{R}^{\ox n} \ox \sigma_{A^n}^u)\\
\leq &D^*_{\alpha}(\mc{E}_n \ox \mc{E}_{\sigma_{A^n}^u}(\rho_{RA}^{\ox n}) \| I_{R}^{\ox n} \ox \sigma_{A^n}^u)+2\log(n+1)^{|R|}+2\log v_{n,|A|} \\
\leq &D^*_{\alpha}(\mc{E}_n(\rho_{RA}^{\ox n}) \| I_{R}^{\ox n} \ox \sigma_{A^n}^u)+2\log(n+1)^{|R|}+2\log v_{n,|A|}.
\end{split}
\end{equation}
Eq.~(\ref{equ:typecon}) gives
\begin{equation}
\label{equ:conlim}
\lim_{n \rightarrow \infty} \frac{D^*_{\alpha}(\mc{E}_n(\rho_{RA}^{\ox n}) \| I_{R}^{\ox n} \ox \sigma_{A^n}^u)}{n}=\inf_{\sigma_A \in \mc{S}(A)}D^*_\alpha(\rho_{RA}\|I_{R}\ox \sigma_A)=-H^*_{\alpha}(R|A)_\rho.
\end{equation}
Recall that
\begin{equation}
D^*_{\alpha}(\mc{E}_n(\rho_{RA}^{\ox n}) \| I_{R}^{\ox n} \ox \sigma_{A^n}^u) 
=\frac{1}{\alpha-1} \log \sum_{t \in \mc{T}_n^\mc{X}} P(t)^\alpha Q^*_\alpha(\frac{\Pi_t\rho_{RA}^{\ox n}\Pi_t}{P(t)}\|\Pi_t\ox \sigma_{A^n}^u).
\end{equation}
We can evaluate $D^*_{\alpha}(\mc{E}_n(\rho_{RA}^{\ox n}) \| I_{R}^{\ox n} \ox \sigma_{A^n}^u)$ as
\begin{equation}
\label{equ:equico}
\begin{split}
&\inf_{t \in \mc{T}_n^\mc{X}} \big\{ \frac{\alpha}{\alpha-1}\log P(t)+D^*_{\alpha}(\frac{\Pi_t\rho_{RA}^{\ox n}\Pi_t}{P(t)}\|\Pi_t\ox \sigma_{A^n}^u) \big\}+\frac{|R|}{\alpha-1} \log(n+1)\\
\leq &D^*_{\alpha}(\mc{E}_n(\rho_{RA}^{\ox n}) \| I_{R}^{\ox n} \ox \sigma_{A^n}^u)  \\
\leq &\inf_{t \in \mc{T}_n^\mc{X}} \big\{ \frac{\alpha}{\alpha-1}\log P(t)+D^*_{\alpha}(\frac{\Pi_t\rho_{RA}^{\ox n}\Pi_t}{P(t)}\|\Pi_t\ox \sigma_{A^n}^u) \big\}
\end{split}
\end{equation}
Eq.~(\ref{equ:confi}) follows from Eq.~(\ref{equ:conlim}) and Eq.~(\ref{equ:equico}). 

Similarly, $H_{\beta}(R^n)_{\rho^{\ox n}}$ can be evaluated as 
\begin{equation}
\label{equ:entropy}
\begin{split}
 &\inf_{t \in \mc{T}_n^\mc{X}} \big\{\frac{\alpha}{\alpha-1}\log P(t)+\log|T_n^t|\big\}+\frac{|R|}{1-\beta} \log(n+1) \\
\leq &H_{\beta}(R^n)_{\rho^{\ox n}} \\
\leq  &\inf_{t \in \mc{T}_n^\mc{X}} \big\{\frac{\alpha}{\alpha-1}\log P(t)+\log|T_n^t|\big\}.
\end{split}
\end{equation}
Eq.~(\ref{equ:entropy}) implies that Eq.~(\ref{equ:confi}) holds.
\end{proof}

The following lemma can be deduced from the same process as~\cite[Lemma~2]{WangWilde2019resource}. It was also implicitly proved in~\cite{LWD2016strong}.

\begin{lemma}
\label{lem:hof}
Let $\rho, \sigma \in \mc{P}(\mc{H})$ and $\tau \in \mc{P}(\mc{H})$, and suppose $\supp(\sigma)\not\perp\supp(\tau)$. Fix $\alpha \in (\frac{1}{2}, 1)$ and $\beta\in(1,+\infty)$
such that $\frac{1}{\alpha}+\frac{1}{\beta}=2$. Then
\begin{equation}
\frac{2\alpha}{1-\alpha} \log F(\rho, \sigma) \leq D^*_\beta(\rho\|\tau)-D^*_\alpha(\sigma\|\tau)+\frac{1}{1-\alpha}\log \tr \sigma+\frac{1}{\beta-1}\log \tr\rho.
\end{equation}
\end{lemma}

The following Lemmas are from~\cite[Lemma A.2]{Tomamichel2012framework}, \cite[Lemma 30]{LiYao2024operational}, \cite[Lemma 31]{LiYao2024operational}, \cite[Lemma 10]{MarioYao2024strong} and~\cite[Lemma 6]{ABJT2020partially}, respectively.
\begin{lemma}
\label{lem:opein}
Let $M_{AB} \in \mc{P}(AB)$. Then, $M_{AB} \leq |A| I_A \ox M_B$.
\end{lemma}
 
\begin{lemma}
\label{lem:appen1}
Let $\rho, \sigma\in\mc{P}(\mc{H})$, and let $\mc{H}=\bigoplus_{i\in \mc{I}}\mc{H}_i$ decompose into a set of mutually orthogonal subspaces $\{\mc{H}_i\}_{i\in \mc{I}}$. Suppose that $\sigma=\sum\limits_{i \in \mc{I}} \sigma_i$ with $\supp(\sigma_i)\subseteq \mc{H}_i$. Then
\begin{equation}
F\left(\sum_{i \in \mc{I}} \Pi_i \rho \Pi_i, \sigma\right) \leq \sqrt{|\mc{I}|} F(\rho, \sigma),
\end{equation}
where $\Pi_i$ is the projection onto $\mc{H}_i$.
\end{lemma}
\begin{remark}
\cite[Lemma 30]{LiYao2024operational} only proves Lemma~\ref{lem:appen1} for quantum states $\rho$ and 
$\sigma$. However, the same arguments can be used to show that Lemma~\ref{lem:appen1} still holds for any $\rho, \sigma\in\mc{P}(\mc{H})$.
\end{remark}

\begin{lemma}
\label{lem:fidelity-re}
Let $\rho, \sigma, \tau \in\mc{S}(\mc{H})$ be any quantum states. Then we have
\begin{equation}\label{eq:fid-re}
-\log F^2(\rho,\sigma) \leq D(\tau\|\rho) + D(\tau\|\sigma).
\end{equation}
\end{lemma}

\begin{lemma}
\label{lem:rela}
For any pure state $\varphi$ and subnormalized state $\rho$, it holds that
\begin{equation}
P(\rho, \varphi) \leq \sqrt{d(\rho, \varphi)}.
\end{equation}
\end{lemma}

\begin{lemma}
\label{lem:monohash}
Let $\rho_{XB}=\sum\limits_{x \in \mc{X}} p_x \proj{x}_X \ox \rho^x_{B}$ be a classical-quantum state, $\epsilon \in [0,1]$, and $f:\mc{X} \rightarrow \mc{Z}$ be a function. Then, we have for $\omega_{ZB}:=\sum\limits_{x \in \mc{X}} p_x \proj{f(x)}_Z\ox \rho_B^x$ that
\begin{equation}
 H^{\epsilon,P}_{\rm{min}}(Z|\dot{B})_\omega  \leq  H^{\epsilon,P}_{\rm{min}}(X|\dot{B})_\rho.
\end{equation}
Moreover, when $B=Y$ is classical then we also have $H^{\epsilon,d}_{\rm{min}}(Z|\dot{Y})_\omega  \leq  H^{\epsilon,d}_{\rm{min}}(X|\dot{Y})_\rho$.
\end{lemma}


\subsection{Comparison between classical and pure states}

In Theorem~\ref{thm:main1}, we derive the strong converse exponents based on trace distance for classical states~(Eq.~(\ref{equ:mainl}) and Eq.~(\ref{equ:clamutual})). They are 
the Legendre transforms of the corresponding Petz R\'enyi conditional entropy and Petz R\'enyi mutual 
information. The readers might conjecture that Eq.~(\ref{equ:mainl}) and Eq.~(\ref{equ:clamutual}) still hold for general quantum states. However, we prove in the following proposition that when $\rho_{RA}$ is a pure state, Eq.~(\ref{equ:mainl}) and Eq.~(\ref{equ:clamutual})  can not be reduced to the formulas of strong converse exponents based on trace distance for pure states~(Eq.~(\ref{equ:main}) and Eq.~(\ref{equ:puremutual})). Hence, this reveals that the strong converse exponents of partially smoothed information measures based on trace distance are not uniform across states.

\begin{proposition}
    For any pure state $\ket{\rho}_{RA}=\sum\limits_{x \in \mc{X}} \sqrt{p(x)} \ket{x}_R\ox \ket{x}_A$  and $r \in \mathbb{R}$, we have
    \begin{align}
    \sup_{0 \leq \alpha \leq 1} (1-\alpha) \{r-\bar{H}_\alpha(A|R)_\rho  \}&=
    \inf_{t \in \mc{Q}(\mc{X})} \{D(t\|p)+|r+H(t)+D(t\|p)|^+ \}, \\
    \sup_{0 \leq \alpha \leq 1} (1-\alpha) \{I_{\alpha}(R:A)_\rho-r  \}&=\sup_{\beta>1}\frac{\beta-1}{\beta} \{2H_{2\beta-1}(R)_\rho-r \}.
    \end{align}
\end{proposition}
\begin{proof}
By duality relation of the Petz R\'enyi conditional entropy~(Proposition~\ref{prop:mainpro}~(\romannumeral3)), we have 
\begin{equation}
\begin{split}
    &\sup_{0 \leq \alpha \leq 1} (1-\alpha) \{r-\bar{H}_\alpha(A|R)_\rho  \}\\
    =&\sup_{0 \leq \alpha \leq 1} (1-\alpha) \{r-D_{2-\alpha}(\rho_A\|I_A)\} \\
    =&\sup_{0 \leq \alpha \leq 1}\inf_{\tau_A \in \mc{S}_{\rho_A}} (1-\alpha) \{r-D(\tau_A\|I_A)+\frac{2-\alpha}{1-\alpha}D(\tau_A\|\rho_A)    \} \\
    =&\sup_{0 \leq \alpha \leq 1}\inf_{\tau_A \in \mc{S}_{\rho_A}} \{D(\tau_A\|\rho_A)+(1-\alpha)(r+H(A)_\tau+D(\tau_A\|\rho_A))   \} \\
    =&\inf_{\tau_A \in \mc{S}_{\rho_A}} \sup_{0 \leq \alpha \leq 1} \{D(\tau_A\|\rho_A)+(1-\alpha)(r+H(A)_\tau+D(\tau_A\|\rho_A))   \} \\
    =&\inf_{\tau_A \in \mc{S}_{\rho_A}} \{D(\tau_A\|\rho_A)+|r+H(A)_\tau+D(\tau_A\|\rho_A)|^+   \} \\
    =& \inf_{t \in \mc{Q}(\mc{X})} \{D(t\|p)+|r+H(t)+D(t\|p)|^+ \},
\end{split}
\end{equation}
where the third line is due to Proposition~\ref{prop:mainpro}~(\romannumeral11) and the fifth line comes from Sion's minimax theorem.

Similarly, we can get from the duality relation of the Petz R\'enyi mutual information
\begin{equation}
\begin{split}
&\sup_{0 \leq \alpha \leq 1} (1-\alpha) \{I_{\alpha}(R:A)_\rho-r  \}\\
=&\sup_{0 < \alpha <1} (1-\alpha) \{I_{\alpha}(R:A)_\rho-r  \}\\
=&\sup_{0 < \alpha < 1} (1-\alpha) \{-D^*_{\frac{1}{\alpha}}(\rho_R\|\rho_R^{-1})-r  \}\\
=&\sup_{0 < \alpha < 1} (1-\alpha) \{2H_{\frac{2}{\alpha}-1}(R)_\rho-r   \} \\
=&\sup_{\beta>1}\frac{\beta-1}{\beta} \{2H_{2\beta-1}(R)_\rho-r \}.
\end{split}
\end{equation}
\end{proof}


\end{document}